\documentclass[journal]{IEEEtran}

\usepackage[pdftex]{graphicx}  

\usepackage{amsmath}
\usepackage{amssymb}
\usepackage{amsthm}

\usepackage{braket}

\usepackage{color}
\usepackage{multirow}
\usepackage{enumerate}

\newtheorem{Lmm}{Lemma}
\newtheorem{Thm}{Theorem}
\newtheorem{Dfn}{Definition}
\newtheorem{Crl}{Corollary}

\begin{document}

\title{Advantage of the key relay protocol over secure network coding}

\author{Go Kato, Mikio Fujiwara, and Toyohiro Tsurumaru%
\thanks{
Go Kato and
Mikio Fujiwara are with NICT, Nukui-kita, Koganei, Tokyo 184-8795, Japan (e-mail: go.kato@nict.go.jp, fujiwara@nict.go.jp).
Toyohiro Tsurumaru is with Mitsubishi Electric Corporation, Information Technology R\&D Center,
5-1-1 Ofuna, Kamakura-shi, Kanagawa, 247-8501, Japan (e-mail: Tsurumaru.Toyohiro@da.MitsubishiElectric.co.jp).
}
}

\maketitle

\begin{abstract}
The key relay protocol (KRP) plays an important role in improving the performance and the security of quantum key distribution (QKD) networks.
On the other hand, there is also an existing research field called secure network coding (SNC), which has similar goal and structure.
We here analyze differences and similarities between the KRP and SNC rigorously.
We found, rather surprisingly, that there is a definite gap in security between the KRP and SNC;
that is, certain KRPs achieve better security than any SNC schemes on the same graph.
We also found that this gap can be closed if we generalize the notion of SNC by adding free public channels;
that is, KRPs are equivalent to SNC schemes augmented with free public channels.
\end{abstract}

\section{Introduction}

The key relay protocol (KRP)  plays an important role in improving the performance and the security of quantum key distribution (QKD) networks \cite{SPDALL10,ALLEAUME201462,10.1007/978-3-540-85093-9_4,itut_y3800}.
On the other hand, there exists another research field called secure network coding (SNC; see, e.g., Refs. \cite{1023595,5592818}), which has the goal and structure similar to the KRP.
The goal of this paper is to analyze differences and similarities between the KRP and SNC rigorously.

QKD realizes distribution of secret keys to players at distant locations (see, e.g., Refs. \cite{RevModPhys.74.145,RevModPhys.92.025002}).
However, the communication distance achievable by a single QKD link is limited by the technological level of quantum optics \cite{RevModPhys.92.025002}.
KRPs are used to enable key distribution beyond such limitation of a single QKD link.
The basic idea of the KRP is to pass a secret key of one QKD link on to another QKD link with the help of insecure public channels, such as the internet (cf. Figs. \ref{fig:key_relay_example1} and \ref{fig:key_relay_example2}).

The KRP has similarities and differences with SNC (Table \ref{table:differences_similarity_KRs_SNCs}).
While they share the same goal of sharing secret messages, they differ in that
1) Public channels are available in KRPs, but not in SNC schemes, 
2) KRPs use QKD links (or more generally, local key sources) while SNC schemes use secret channels, and
3) The messages in KRPs must be a random bit, while in SNC schemes each sender can freely choose its message.

Then the question naturally arises whether these differences are really essential.
For example, is it not possible that there is actually a way of converting KRPs to SNC schemes, and that they are shown to be equivalent?
The goal of this paper is to answer to this question.
For the sake of simplicity, we will limit ourselves to the one-shot scenario.

The outline of our results is as follows (Fig. \ref{fig:inclusion_relation}).

If we generalize SNC \cite{1023595} by adding public channels (see the third column of Table \ref{table:differences_similarity_KRs_SNCs}), then KRPs and SNC schemes (with public channels) on the same graph are always equivalent (Theorem \ref{thm:SNC_and_KR}).

On the other hand, if we do not generalize SNC and limit ourselves to its conventional form (without public channels; see the second column of Table \ref{table:differences_similarity_KRs_SNCs}), then there is a definite gap in security between the KRP and SNC:
On some graphs a KRP achieves better security than any conventional SNC schemes (Theorem \ref{crl:SNCs_and_KRs_are_not_equivalent} and Corollary \ref{tmp:crl:SNCs_and_KRs_are_not_equivalent}).
Hence the accumulation of past research on the conventional SNC is not sufficient to explore the potential of KRPs.
This suggests that the KRP is a new research field.

\begin{table*}[htb]
\label{table:differences_similarity_KRs_SNCs}
\caption{Similarities and differences between the KRP, the conventional SNC, SNC with public channels, and KRP-by-SNC.}
\begin{center}
\begin{tabular}{c|cccc}
&\multirow{2}{*}{Key relay protocol}
 & Secure network coding (SNC) & &  \\
&\multirow{2}{*}{(KRP)}& without public channels & SNC with public channels & KRP-by-SNC\\
 &  & (Conventional SNC)& & \\
\hline
\multirow{2}{*}{Public channels} & \multirow{2}{*}{\checkmark} & & \multirow{2}{*}{\checkmark} & \\
&&&&\\
Local key sources & \multirow{2}{*}{\checkmark}&  & & \\ 
(e.g., QKD links) &&\\
\multirow{2}{*}{Secret channels} & & \multirow{2}{*}{\checkmark} & \multirow{2}{*}{\checkmark} & \multirow{2}{*}{\checkmark}\\
&&&& \\
\hline
\multirow{2}{*}{Goal}&\multicolumn{4}{c}{\multirow{2}{*}{Each sender-receiver pair (or each user pair) $i$ share a secret message}}\\
&&&& 
\\
\hline
\multirow{2}{*}{Content of the message }& \multirow{2}{*}{Random bit $k_i$} & Message $m_i$ chosen & Message $m_i$ chosen  & \multirow{2}{*}{Random bit $k_i$}\\
& & by the sender & by the sender & 
\end{tabular}
\end{center}
\end{table*}

\begin{figure}[htbp]
\begin{center}
\includegraphics[bb=0 0 960 440, width=\linewidth, clip]{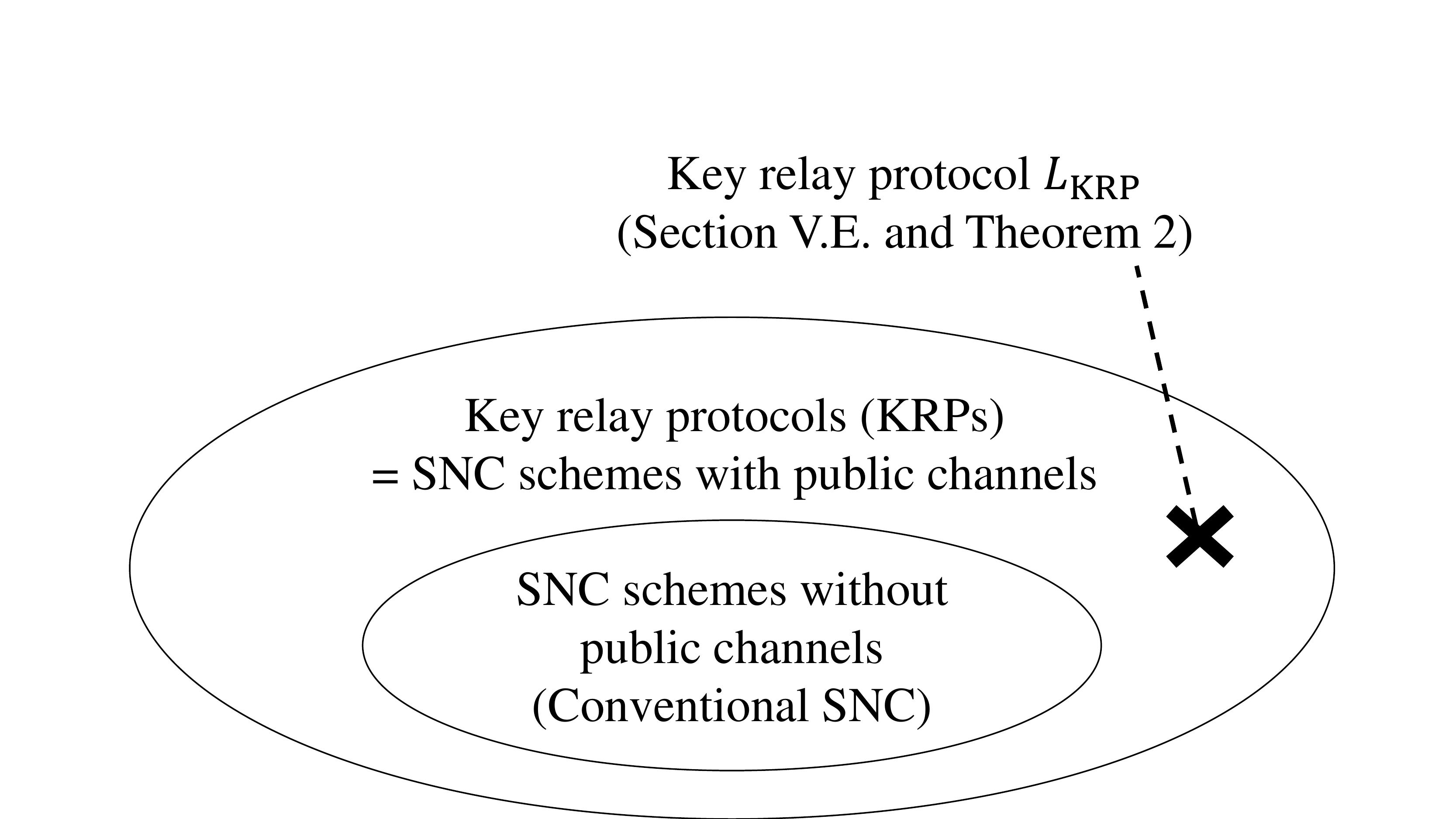}
\end{center}
 \caption{
Relation of secure network coding (SNC) with and without public channels, and the key relay protocol (KRP).
The settings and the goals of SNC and the KRP are summarized in Table \ref{table:differences_similarity_KRs_SNCs}.
The KRP and SNC with publich channels always achieve the same level of security (Theorem \ref{thm:SNC_and_KR}).
The security of KRP is better than that of the conventional SNC or SNC without public channels (Theorem \ref{crl:SNCs_and_KRs_are_not_equivalent}).
}
 \label{fig:inclusion_relation}
\end{figure}

\section{Key relay protocol (KRP)}

\subsection{Motivation and examples of the KRP}

Quantum key distribution (QKD) distributes secret keys to two separate players.
However, the communication distance achievable by a single set of QKD devices, or a {\it QKD link}, is limited by the technological level of quantum optics, and is currently in the order of 100 km \cite{RevModPhys.92.025002}.
For this reason, in this paper, we refer to a QKD link also as a {\it local key source}. 

There is of course a strong demand to distribute secret keys globally, or beyond the reach of a single QKD link.
The {\it key relay protocol} (KRP) \cite{SPDALL10,ALLEAUME201462,10.1007/978-3-540-85093-9_4} aims to fulfill this demand by connecting multiple QKD links, and also by using insecure public channels, such as the internet.

Fig. \ref{fig:key_relay_example1} illustrates the simplest example of such KRPs.
Users $u^1$ and $u^2$ are separated by twice the reach of a local key source, and are connected by two local key sources $LKS_{e_1}$ and $LKS_{e_2}$.
From these local key sources, users $u^1$ and $u^2$ receive distinct local keys $r_{e_1}, r_{e_2}\in_{\rm R}\{0,1\}$ respectively.
Then, in order for both $u^1$ and $u^2$ to share the same key $k^1=k^2$, which we call the {\it relayed key}, 
they execute the following procedure with the help of the midpoint $v$:
\begin{enumerate}
\item The midpoint $v$ announces the difference of the two local keys,  $\Delta r=r_{e_1}+r_{e_2}$.
\item Users $u^1$ and $u^2$ calculate the relayed keys $k^1=r_{e_1}$ and  $k^2=r_{e_2}+\Delta r$, respectively.
\end{enumerate}
Note that $k^1=k^2$ is indeed satisfied.
Note also that $k_i$ remain secret even if the announcement $\Delta r$ is revealed.

This idea can be generalized to more complex network configurations.
For example, one can improve the distance by serially extending the above construction (Fig. \ref{fig:key_relay_example2}(a)), or can improve the security by extending it in parallel (Fig. \ref{fig:key_relay_example2}(b)).
In the next subsection, we will give a formal definition of KRPs, applicable to an arbitrary network configuration.

\begin{figure}[htbp]
\begin{center}
 \includegraphics[bb=0 0 950 370, width=\linewidth, clip]{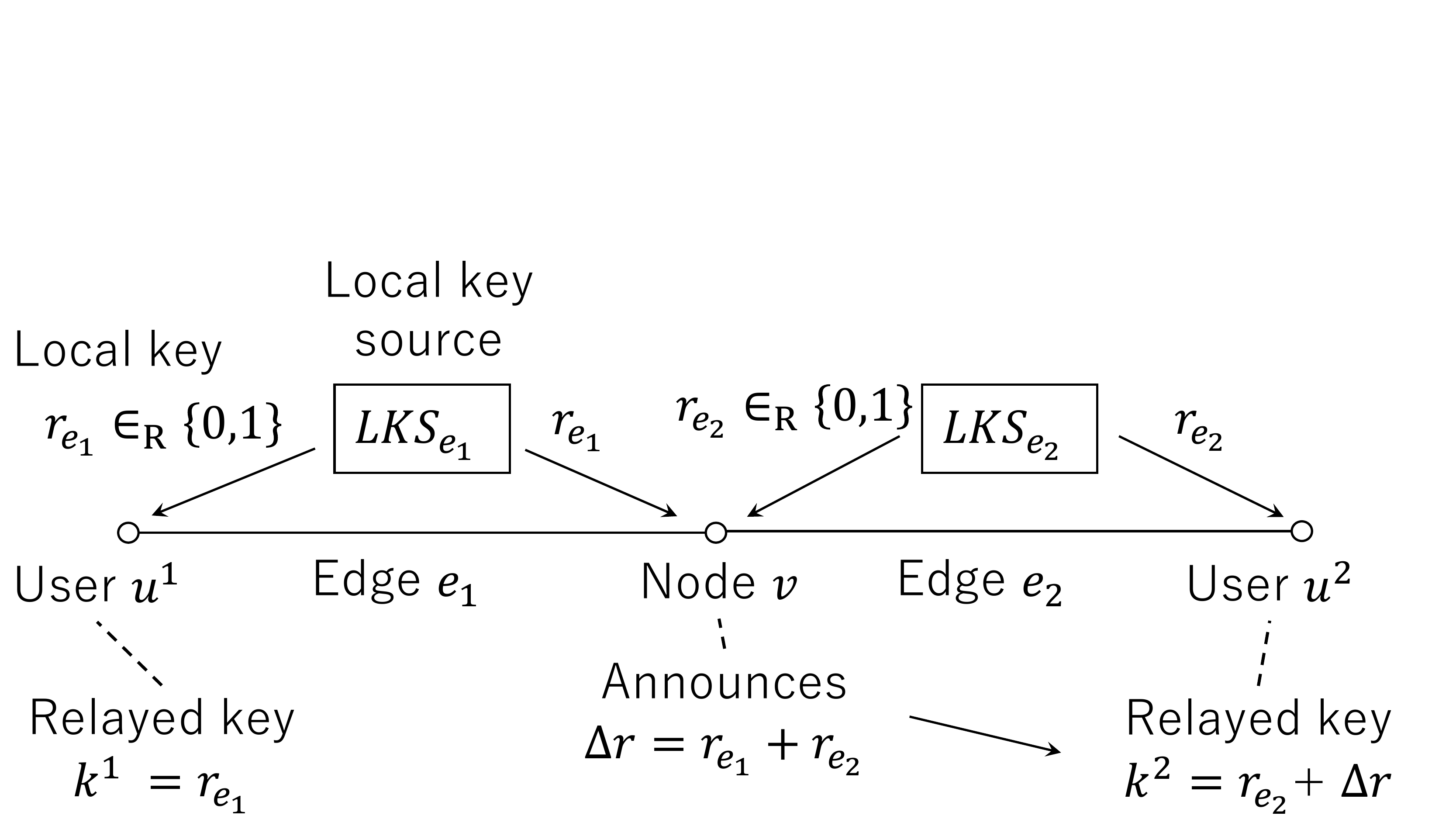}
\end{center}
 \caption{
The simplest example of the KRP.
On each edge $e_i$ there is a local key source $LKS_{e_i}$ which distributes a random bit $r_{e_i}\in_{\rm R}\{0,1\}$ to both ends.
Each node can also use public channels freely.
User pair $u^1,u^2$ wishes to share a relayed key $k=(k^1,k^2)$.
To this end, the midpoint $v$ announces $\Delta r=r_{e_1}+r_{e_2}$, and then user $u^1$ and $u^2$ each calculate $k^1=r_{e_1}$ and $k^2=r_{e_2}+\Delta r$.
}
 \label{fig:key_relay_example1}
\end{figure}

\begin{figure}[htbp]
\begin{center}
\includegraphics[bb=0 0 950 550, width=\linewidth, clip]{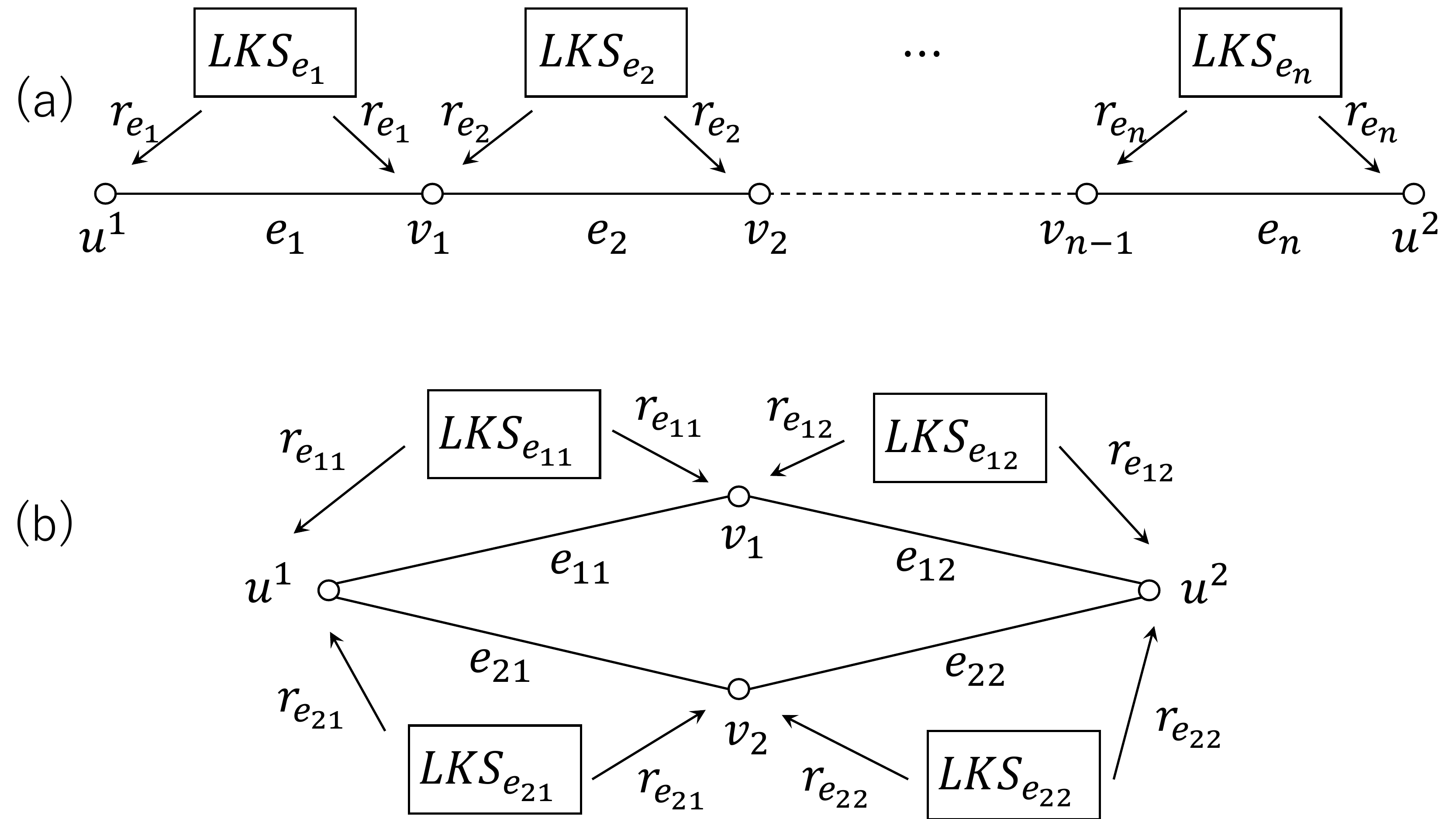}
\end{center}
 \caption{
Somewhat complex examples of the KRP.
(a) Serialization of Fig. \ref{fig:key_relay_example1}.
Nodes $v_i$ each announce $\Delta r_i=r_i+r_{i+1}$, and then users $u^1$ and  $u^2$ calculate relayed keys $k^1=r_{e_1}$ and $k^2=r_n+\sum_{i=1}^{n-1}\Delta r_i$ respectively.
(b) A parallelization of Fig. \ref{fig:key_relay_example1}.
Nodes $v_i$ each announce $\Delta r_i=r_{e_{i1}}+r_{e_{i2}}$, and then users $u^1,u^2$ each calculate $k^1=r_{e_{11}}+r_{e_{21}}$, $k^2=\sum_{i=1,2}(r_{e_{2i}}+\Delta r_i)$.
Note that the relayed key $k=(k^1,k^2)$ remains secret here even if someone takes over an edge set $E_i=\{e_{i1},e_{i2}\}$ ($i=1$ or 2) and leaks  local keys $r_{e_{i1}},r_{e_{i2}}$.
In this sense we regard this construction more secure than that of Fig. \ref{fig:key_relay_example1}.
}
 \label{fig:key_relay_example2}
\end{figure}

\subsection{Formal definition of the KRP}
\label{sec:formal_definition_of_KRP}

The outline is that:
On an undirected graph $G=(V,E)$, pairs of users wish to share a relayed key with the help of other players on nodes $V$ having access to local key sources and a public channels, without disseminating the message to the adversary.

\subsubsection{Setting}
\label{sec:setting_KR}
An undirected graph $G=(V,E)$ consists of a node set $V$ and an edge set $E$.
For the sake of simplicity, we assume that $G$ are connected.
Each node $v\in V$ has an individual player (denoted by the same symbol as the node), some of which constitute $n_{\rm pair}$ pairs of users $u_i=(u^1_i,u^2_i)$ with $i=1,\dots,n_{\rm pair}$.
There is also an adversary, who can wiretap some edges $\subset E$.

Each edge $e\in E$ has a local key source $LKS_e$ and a public channel $PC_e$, which behave as follows.

\begin{Dfn}[Local key sources and public channels]
\label{dfn:random_source}
$LKS_e$ and $PC_e$ operate as follows:
\begin{itemize}
\item Local key source $LKS_e$ (Fig. \ref{fig:random_source} (a)): On input ``start'' command from an end node $v$ or $w$, it sends a local key, or a uniformly random bit  $r_{e}\in_{\rm R}\{0,1\}$ to both $v$ and $w$.
When edge $e$ is wiretapped, it also sends $r_e$ to the eavesdropper.
\item Public channel $PC_e$ (Fig. \ref{fig:random_source} (b)): On input a bit string $p_e\in\{0,1\}^*$ from an end node (say, $v$), it sends $p_e$ to the other end node (say, $w$) and to the adversary.
\end{itemize}
\end{Dfn}

\begin{figure}[htbp]
\begin{center}
\includegraphics[width=\linewidth, clip]{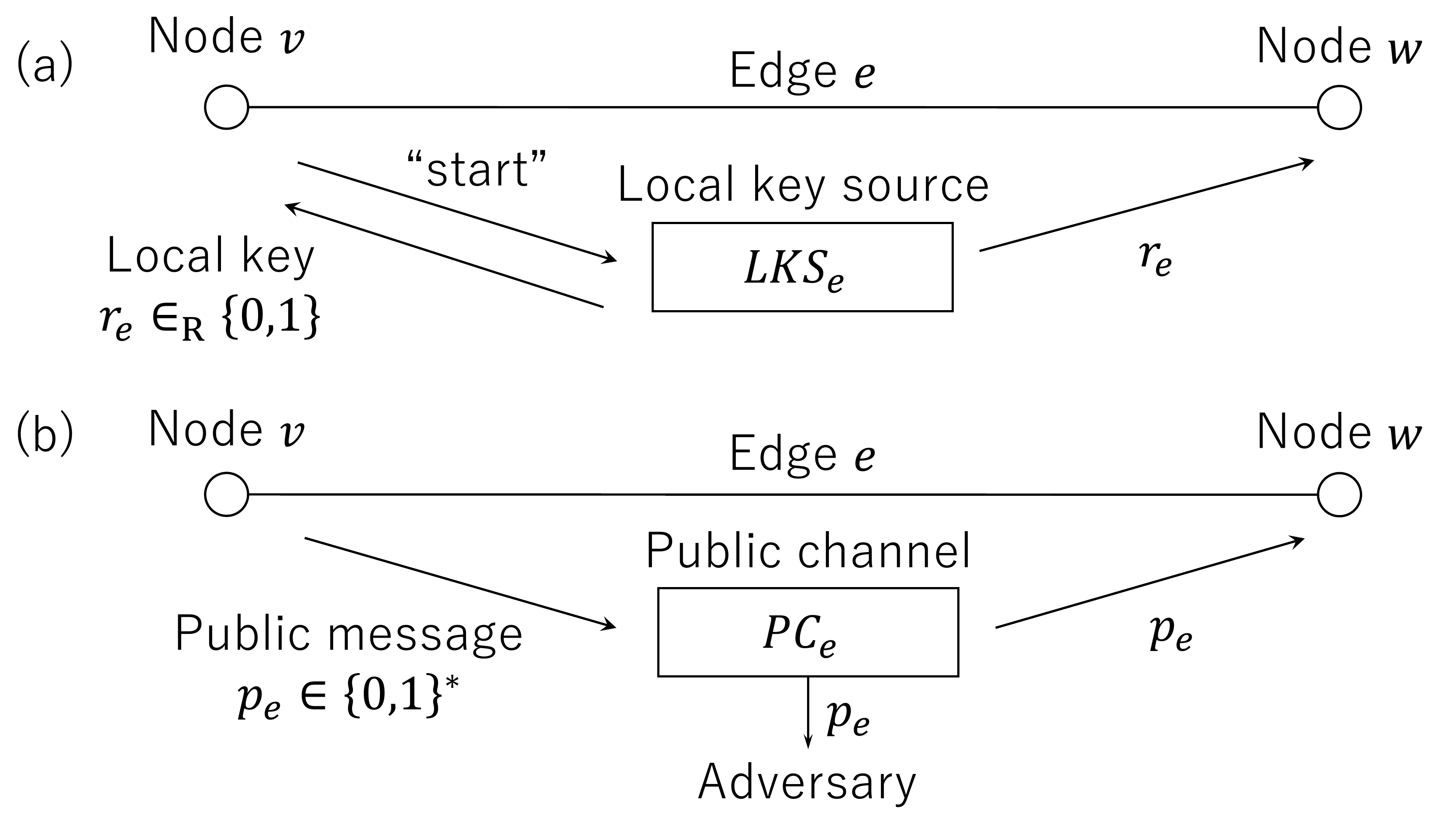}
\end{center}
 \caption{(a) Behavior of local key source $LKS_e$ in the absence of the adversary, on edge $e$ having end nodes $v,w$, (b) public channel $PC_e$ on the same edge.}
 \label{fig:random_source}
\end{figure}

\subsubsection{Key relay protocol}
\label{sec:construction_key_relays}

With the above setting, each user pair $u_i=(u^1_i,u^2_i)$ wish to share a relayed key $k_i=(k^1_i,k^2_i)$ with the help of players $V$, without disseminating $k_i$ to the adversary.
To this end, they request all nodes $V$ to execute a procedure of the
the following type.

\begin{Dfn}
\label{dfn:key_relay}
A protocol $L$ of the following type, performed by players $V$, is called a key relay protocol (KRP).
\begin{enumerate}
\item All players $V$ communicate using public channels $PC_e$ and local key sources $LKS_e$\footnote{More precisely, the outputs ($p_e$, $r_e$, or ``start'') of players $V$ are defined as functions of previously received data ($\subset\{p_e, r_e| e\in E\}$) and of random variables generated by the player.
Each player sends out the outputs whenever necessary data are all received.}.

Here each $LKS_e$ can only be used once, while $PC_e$ can be used arbitrarily many times.
\item Each user $u_i^j$ calculates a relayed key $k_i^j$.
\end{enumerate}
\end{Dfn}

\subsubsection{Security criteria}
\label{sec:security_criteria}
There is a known collection ${\cal E}^{\rm adv}=\{E^{\rm adv}_{1}, E^{\rm adv}_{2},\dots\}$ of edge set $E^{\rm adv}_{i} \subset E$ which the adversary can wiretap on.
In each round of the protocol, the adversary chooses $E^{\rm adv}_{l}\in {\cal E}^{\rm adv}$ and wiretaps edges $e\in E^{\rm adv}_{l}$.

\begin{Dfn}[Security of the KRP]
\label{dfn:security_criteria}
A key relay protocol $L$ is secure against ${\cal E}^{\rm adv}$, if it satisfies the followings.
\begin{itemize}
\item Soundness: 
The relayed keys $k_i^1,k_i^2$ generated by user pair $u_i=(u^1_i,u^2_i)$ are equal and uniformly distributed; i.e., $\Pr[K_i^1=K_i^2]=1$, and $\Pr[K_i^j=0]=\Pr[K_i^j=1]=1/2$.

Also, $k_i^j$ generated by different user pairs are independent.

\item Secrecy: The relayed key pairs $k_i=(k_i^1,k_i^2)$ are unknown to the adversary even when any edge set $E_l^{\rm adv}\in {\cal E}^{\rm adv}$ is wiretapped.
That is, for any $l$, we have
\begin{equation}
I(K_1,K_2,\dots,K_{\rm n_{\rm pair}}: A(E_l^{\rm adv}))=0,
\end{equation}
where $A(E_l^{\rm adv})$ denotes the information that the adversary obtains by eavesdropping on edge set $E_l^{\rm adv}$.
\end{itemize}
\end{Dfn}
$A(E_l^{\rm adv})$ appearing in Definition \ref {dfn:security_criteria} consists of local keys $r_e$ on edges $e\in E^{\rm adv}_{l}$, and of all public information $p_e$ ($e\in E$).


\subsection{Notes on KRPs used in practical QKD networks}

In fact, the KRP defined above is slightly different from those used in actual QKD networks.
Below we elaborate on their relation.

\subsubsection{Edge adversary model vs. node adversary model}
In  the above definition, we employed the edge adversary model (where the adversary eavesdrop on some edges), while in actual QKD networks the node adversary model (where the adversary can eavesdrop on information that goes in and out of a certain edges set) is usually assumed.
This is not really a limitation, since the former model incorporates the latter: The situation where ``the adversary eavesdrop on a node $v$'' in the node adversary model can always be described as ``all edges surrounding $v$ are wiretapped'' in the edge adversary model.

\subsubsection{Passive adversary vs active adversary}
Above we assumed that the adversary is passive (honest but curious), meaning that she eavesdrops on, but does not tamper with communication.
On the other hand, in QKD, one usually assumes that the adversary  is active; i.e., she can both eavesdrop on and tamper with communication.

The easiest way to convince oneself of this limitation, of course, is to accept it merely as a simplification introduced at the first step of continuing research.

On the other hand, there are also ways of justifying this limitation to some extent.
That is, if the adversary is active, 
the following two problems arise,
\begin{itemize}
\item Problem with soundness: The relayed keys may not match, $\Pr[k_i^1\ne k_i^2]>0$.
\item Problem with secrecy:
Players $V$ may malfunction and leak extra information to the adversary, damaging the secrecy.
\end{itemize}
but, in practical QKD networks, there are ways to solve or work around both these problems.

\paragraph{How to work around the problem with soundness}
The basic idea here is the following.
The relayed keys $k_i=(k_i^1, k_i^2)$ are random bits and are not meaningful by themselves, and thus can be discarded at any time.
Hence, even if the event $k_i^1\ne k_i^2$ occurs, players can discard $k_i^1,k_i^2$ and repeat new rounds the KRP (including QKD as local key sources) until they obtain $k_i^1,k_i^2$ satisfying $k_i^1=k_i^2$.
This can generally decrease the key generation speed, but the secrecy remains intact. 

Of course, in order for the above idea to actually function in practice, user pairs $u_i$ must be able to detect an error  (check if $k_i^1=k_i^2$ or not) with a sufficiently small failure probability.
This is also realizable by using information-theoretically secure message authentication codes (see, e.g., Section 4.6 of Ref. \cite{KatzLindell}).

Combining these ideas, we obtain the following method.
\begin{enumerate}
\item User pairs $u_i=(u^1_i,u^2_i)$ repeat a KRP $n$ times and share $n$-bit relayed keys $\vec{k_i^1},\vec{k_i^2}\in\{0,1\}^n$.
\item User $u_i^1$ calculates the hash value $\sigma_i=h(\vec{k_i^1})$ of $\vec{k_i^1}$ using an $\varepsilon$-difference universal hash function $h$ \cite{KatzLindell}.
User $u_i^1$ then encrypts $\sigma_i$ by the one-time pad scheme (see, e.g., Ref. \cite{KatzLindell}) and sends it to $u_i^2$.
(In fact, this entire step corresponds to authenticating message $\vec{k}^1_i$ using Construction 4.24 of Ref. \cite{KatzLindell}.)
\item User $u_i^2$ decrypts the received ciphertext to obtain $\sigma_i$.
If $\sigma_i\ne h(\vec{k_i^2})$, $u_i^2$ announces that the relayed keys $\vec{k_i^1},\vec{k_i^2}$ must be discarded.
(Here, $u_i^2$ authenticates his announcement by again using Construction 4.24 of Ref. \cite{KatzLindell}.)
\end{enumerate}
In this method, steps 2 and 3 each consume a pre-shared key\footnote{The security proofs of QKD require that its public communication be authenticated.
A customary way to fulfill this requirement in practical QKD systems is that each user pair always keeps sharing a relatively small amount of secret key (pre-shared key), and uses it to authenticate their public communication, e.g., by the methods given in Ref. \cite{WEGMAN1981265} and in Section 4.6, Ref. \cite{KatzLindell}.
Here we use those pre-shared keys also for KRPs.} of a length proportional to $|\sigma_i|$, the length of $\sigma_i$.
However, one can set $|\sigma_i|$ negligibly small compared with $n$, with an appropriate choice of the function $h$ and for sufficiently large $n$.
Thus the net relayed key obtained by this method almost equals $n$.
For example, by using a polynomial-based $\varepsilon$-difference universal hash function, we have $|\sigma_i|=O(\varepsilon^{-1}\log n)$ with $\varepsilon$ being the failure probability of the error detection.

\paragraph{Countermeasure against problem with secrecy}
As for the problem with secrecy, one countermeasure is to restrict ourselves with {\it linear} KRPs.

Here a linear KRP means the one where players $V$ are {\it linear}.
A player $v\in V$ being linear means that its outputs $p_e,r_e$ are all linear functions of previously received data ($\subset\{p_e, r_e| e\in E\}$) and of random variables generated by the player.
In such restricted case we can prove the following lemma.
\begin{Lmm}
\label{lmm:linear_case}
If a linear KRP is secure against passive (i.e., honest but curious) adversaries, it is also secure against active adversaries.
\end{Lmm}
This lemma is a variant of Theorem 1, Ref. \cite{HOKC20}, which was previously obtained for the secure network coding (SNC).
As the proof is essentially the same as in Ref. \cite{HOKC20}, we here only give a sketch:
Suppose for example that the active adversary modifies a local key $r_{e'}$ to $r_{e'}+\Delta r$, which is to be input to a node $v$.
With $v$ being linear, $v$'s subsequent outputs all change linearly in $\Delta r$; for example, a public message $p_e$, which $v$ outputs, changes to $p_e+f(\Delta r)$ with $f$ being a linear function.
Since those linear response to tampering, such as $f(\Delta r)$, are all predictable, we can conclude that the adversary gains nothing by tampering with communication.

\section{Main results: Relation between the KRP and secure network coding (SNC)}
\label{sec:main_result}

As readers familiar with secure network coding (SNC; see, e.g., Refs. \cite{1023595,5592818}) may have already noticed, the KRP defined in the previous section have similarities and differences with SNC (Table \ref{table:differences_similarity_KRs_SNCs}).
That is, while they both share the same goal that each sender-receiver pair (or each user pair) share a secret message, they differ in that
\begin{enumerate}
\item Public channels $PC_e$ are used in the KRP, but not in SNC.
\item The KRP uses local key sources $LKS_e$, while SNC uses secret channels.
\item In SNC, the sender can choose the massage freely.
However, in the KRP, the message (which we called the relayed key $k_i$ in the previous section) must be uniformly random, and thus the sender does not have freedom to choose it.
\end{enumerate}
From this observation the question naturally arises whether these differences are really essential.
For example, is it not possible that there is actually a way of converting KRPs to SNC schemes, and that they are shown to be equivalent?
In this section we answer to this question.
The outline of our results is as follows.

First, if we eliminate difference 1) above by hand, that is, if we generalize SNC \cite{1023595} by adding public channels, then we can simultaneously resolve the remaining differences, items 2) and 3), as well.
As a result of this, we can show that the generalized form of SNC (i.e., SNC with public channels, in the third column of Table \ref{table:differences_similarity_KRs_SNCs}) and the KRP are equivalent (Theorem \ref{thm:SNC_and_KR}). 

On the other hand, if we do not generalize SNC and limit ourselves with its conventional form (the second column of Table \ref{table:differences_similarity_KRs_SNCs}), then there is a definite gap in security between SNC and the KRP:
There are situations where KRPs achieve better securities than the conventional SNC schemes, without public channels (Theorem \ref{crl:SNCs_and_KRs_are_not_equivalent}).

\subsection{Definition of SNC with public channels}
\label{sec:SNC_with_public_channels}

We begin with a formal definition of SNC with public channels, which is mentioned above and corresponds to the third column of Table \ref{table:differences_similarity_KRs_SNCs}.

The conventional SNC (the second column of Table \ref{table:differences_similarity_KRs_SNCs}) is the special case of this scheme where the use of public channels is prohibited.

\subsubsection{Setting}
\label{sec:setting_SNC}
The setting is the same as that of the KRP, given in Section \ref{sec:setting_KR}, except
\begin{itemize}
\item Of each user pair $u_i=(u_i^1,u_i^2)$, one user (say, $u_i^1$) is named the sender $a_i$, and the other (say, $u_i^2$) the receiver $b_i$.
Thus either $(u_i^1,u_i^2)=(a_i,b_i)$ or $(u_i^1,u_i^2)=(b_i,a_i)$ holds.

The is necessary because, in SNC, messages are not a random bit (as in the KRP), but must be chosen by the sender $a_i$; see Definition \ref{dfn:SNC} below.
\item Local key sources $LKS_e$ are replaced by the secret channels $SC_e$, defined in Definition \ref{dfn:channels} below.
\end{itemize}

\begin{Dfn}[Secret channels]
\label{dfn:channels}
On input a bit $s_e\in\{0,1\}$ from one end node (say, $v$), secret channel $SC_e$ sends $s_e$ to the other end node (say, $w$); see Fig. \ref{fig:channels}.
When edge $e$ is wiretapped, it also sends $s_e$ to the eavesdropper.
\end{Dfn}

In comparison with the conventional SNC \cite{1023595}, the setting above differs only in that players $V$ can use public channels $PC_e$ in addition to secret channels $SC_e$ (see Table \ref{table:differences_similarity_KRs_SNCs}).



\begin{figure}[htbp]
\begin{center}
\includegraphics[bb=0 0 1000 200, width=\linewidth, clip]{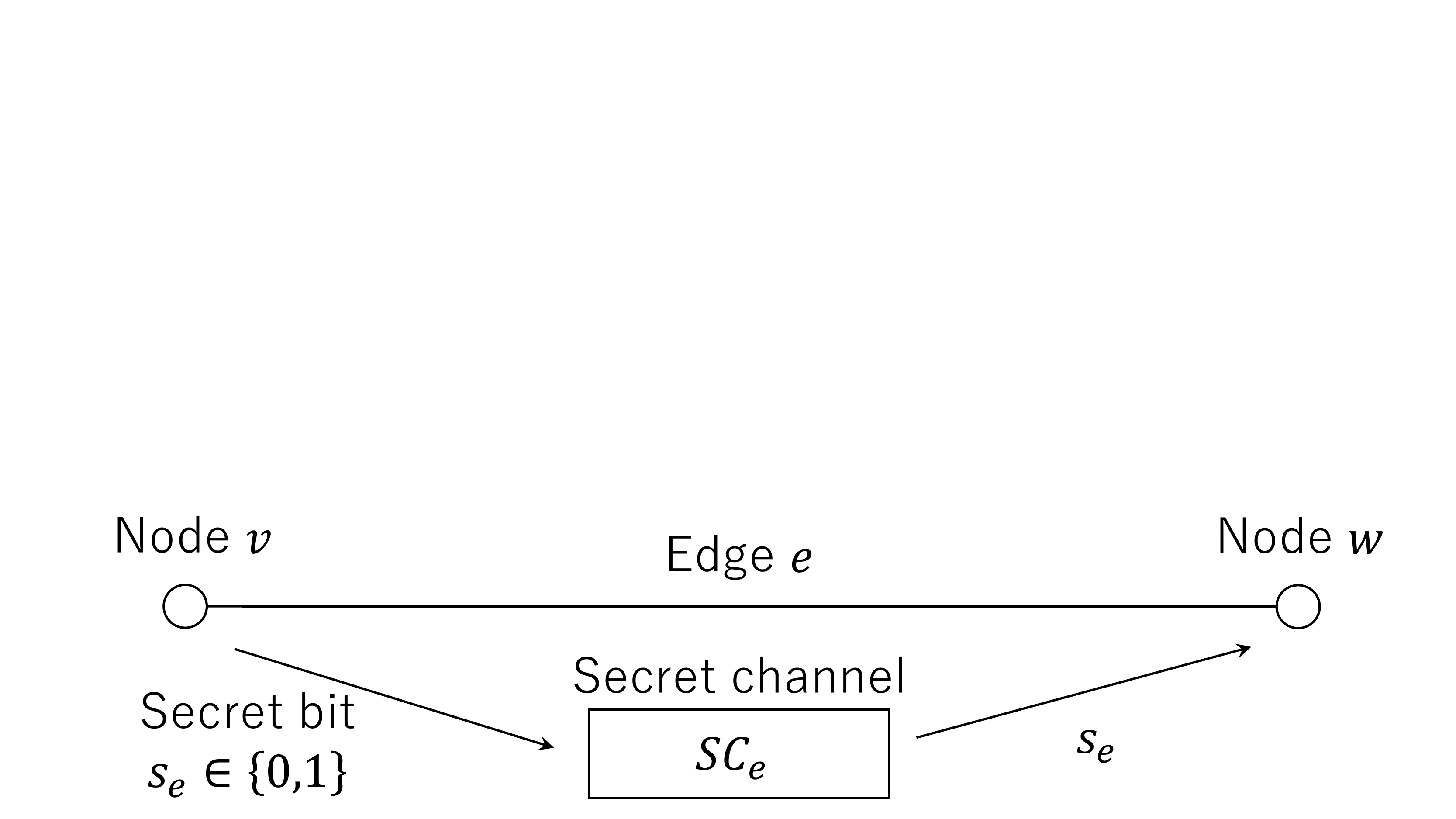}
\end{center}
 \caption{
Behavior of secret channel $SC_e$ in the absence of the adversary.
}
 \label{fig:channels}
\end{figure}

\subsubsection{SNC with public channels}
The goal of our SNC with public channels is the same as that of the conventional SNC \cite{1023595}: Each sender-receiver pair $(a_i,b_i)$ wish to exchange message $m_i$ with the help of other players on nodes $V$, without disseminating $m_i$ to the adversary.
 
\begin{Dfn}[SNC with public channels]
\label{dfn:SNC}
We call a protocol of the following type a secure network coding (SNC) scheme with public channels.
\begin{itemize}
\item Each sender $a_i$ chooses a message $m_i\in\{0,1\}$ aimed at the receiver $b_i$.
\item Players $V$ communicate by using public channels $PC_e$ and secret channels $SC_e$\footnote{As in the case of the KRP, we assume that the outputs ($p_e$, $s_e$) of players are defined as functions of previously received data ($\subset\{p_e, s_e| e\in E\}$) and of random variables generated by the player.
We also assume that each player sends out the output whenever necessary data are all received.}.

Here, each $SC_e$ can only be used once, while $PC_e$ can be used arbitrarily many times.
\item Each receiver $b_i$ calculates message $\hat{m}_i\in\{0,1\}$.
\end{itemize}
\end{Dfn}

In comparison with Definition \ref{dfn:key_relay} for the KRP, Definition \ref{dfn:SNC} above differs only in that $LKS_e$ are replaced by $SC_e$, and that senders $a_i$ can arbitrarily choose message $m_i$, which need not be uniformly distributed, unlike the relayed key $k_i^1$ (cf. Table \ref{table:differences_similarity_KRs_SNCs}).

\subsubsection{Security criteria}
\label{sec:security_criteria_SNC}
The security criteria is essentially the same as Definition \ref{dfn:security_criteria} for the case of the KRP.
That is, there is again a known collection ${\cal E}^{\rm adv}=\{E^{\rm adv}_{1}, E^{\rm adv}_{2},\dots\}$ of wiretap sets $E^{\rm adv}_{l} \subset E$.
In each round of the SNC scheme, the adversary chooses $E^{\rm adv}_{l}\in {\cal E}^{\rm adv}$ and wiretap edges $e\in E^{\rm adv}_{l}$.
\begin{Dfn}[Security of SNC with public channels]
\label{dfn:security_criteria_SNC}
A SNC scheme $L$ is secure against ${\cal E}^{\rm adv}$, if it satisfies the followings.
\begin{itemize}
\item Soundness: 
Sender $a_i$'s message $m_i$ reaches receiver $b_i$ without error; $\Pr[M_i=\hat{M}_i]=1$.
\item Secrecy: Messages $m_i,\hat{m}_i$ are unknown to the adversary even when any edge set $E_j^{\rm adv}\in{\cal E}^{\rm adv}$ is wiretapped.
That is, for any $l$, we have
\begin{equation}
I(M_1,M_2,\dots,M_{n_{\rm pair}}: A(E_l^{\rm adv}))=0,
\end{equation}
where $A(E_l^{\rm adv})$ denotes the information that the adversary obtains by eavesdropping on edges $E_l^{\rm adv}$; i.e., $A(E_l^{\rm adv})$ consists of secret bits $s_e$ on edges $e\in E^{\rm adv}_{l}$, and of all public information $p_e$ ($e\in E$).
\end{itemize}
\end{Dfn}

In comparison with Definition \ref{dfn:security_criteria} for the KRP, Definition \ref{dfn:security_criteria_SNC} above differs in that $m_i$ need not be uniformly distributed (cf. Table \ref{table:differences_similarity_KRs_SNCs}), and that local keys $r_e$ included in the adversary's information $A(E_j^{\rm adv})$ are replaced by secret bits $s_e$.


\subsection{SNC with public channels and the KRP are equivalent}

SNC with public channels thus defined are in fact equivalent to the KRP defined in the previous section.

\begin{Thm}[The security of KRP $=$ The security of SNC with public channels]
\label{thm:SNC_and_KR}
KRPs and SNC schemes with public channels can always achieve the same security.
That is,
\begin{enumerate}
\item Given a KRP $L$ compatible with a graph $G$ and user configuration $u_i=(u_i^1,u_i^2)$ which is secure against wiretap sets ${\cal E}^{\rm adv}$, one can construct a SNC scheme with public channels $L'$ which is compatible with the same $G$ and $u_i$, and also secure against the same ${\cal E}^{\rm adv}$.

This is true whether the sender and the receiver for each user pair $u_i$ in $L'$ are assigned as $(u_i^1,u_i^2)=(a_i,b_i)$ or $(u_i^1,u_i^2)=(b_i,a_i)$  (for the meaning of this notation, see Section \ref{sec:setting_SNC}).
\item Given a SNC scheme $L$ (with or without public channels) compatible with a graph $G$ and a sender-receiver configutation $u_i=(a_i,b_i)$ which is secure against ${\cal E}^{\rm adv}$, one can construct a KRP  $L'$ compatible with the same $G$ and $u_i$, which is secure against the same ${\cal E}^{\rm adv}$.
\end{enumerate}
\end{Thm}
Therefore, if one wishes to analyze the potential and limitations of the KRP, it is necessary and sufficient to investigate SNC with public channels. 

We will prove Theorem \ref{thm:SNC_and_KR} in Section \ref{sec:proof_first_theorem}.

\subsection{SNC without public channels and the KRP are not equivalent}
However, in order for Theorem \ref{thm:SNC_and_KR} above to hold, it was in fact essential that we generalized SNC by adding  public channels.
The equivalence with the KRP no longer holds if we limit ourselves with the conventional SNC, i.e. SNC schemes without public channels.
More precisely, we have the following theorem.
{
\begin{Thm}[The security of KRP $\ne$ The security of SNC without public channels (conventional SNC)]
\label{crl:SNCs_and_KRs_are_not_equivalent}
There exists a conbination of a graph $G$, a user configuration $u_i$, and wiretap sets ${\cal E}^{{\rm adv},G_0}$ for which there exists a secure KRP $L_{\rm KRP}$, but there exists no secure SNC scheme without public channels.

This is true whether the sender and the receiver for each user pair $u_i$ (in the SNC without public channels) are assigned as $(u_i^1,u_i^2)=(a_i,b_i)$ or $(u_i^1,u_i^2)=(b_i,a_i)$  (for the meaning of this notation, see Section \ref{sec:setting_SNC}).
\footnote{It is important to note that this theorem applies even to SNC schemes on {\it un}-directed graphs. 
 If one is somehow allowed to limit oneself with SNC on directed graphs, the counterexample $L_{\rm KRP}$ can be constructed straightforwardly. 
 }
\end{Thm}
The proof of this theorem is give in Section \ref{sec:Proof_of_Theorem2}.

In short, there are situations where the KRPs achieve better securities than the conventional SNC.
Hence the accumulation of past research on the conventional SNC is not sufficient to explore the potential of the KRP.
In this sense, the KRP is a new research field. 

Combining Theorem \ref{thm:SNC_and_KR} and Theorem \ref{crl:SNCs_and_KRs_are_not_equivalent}, we can also obtain the following corollary.

\begin{Crl}[The security of SNC with public channels $\ne$ The security of SNC without public channels (conventional SNC)]
\label{tmp:crl:SNCs_and_KRs_are_not_equivalent}
 There exists a combination of a graph $G$, a user configuration $u_i$, and wiretap sets $E^{\rm adv}$ for which there exists a secure SNC scheme with public channel, but there exists no secure SNC scheme without public channels.
\end{Crl}

\section{Proof of Theorem \ref{thm:SNC_and_KR}}
\label{sec:proof_first_theorem}
To prove item 1), note that operations of $LKS_e$ can be simulated by using $SC_e$.
That is, if an end node $v$ of edge $e$ wishes to send a local key $r_e$ to the other end node $w$, it suffices that $v$ generates a random bit $r_e\in_{\rm R}\{0,1\}$ by itself and sends it to $w$ via $SC_e$ (Fig. \ref{fig:equivalence_model_kr}).

By applying this simulation to all $LKS_e$ included in $L$, one obtains a protocol $L'$ where user pairs $u_i=(u_i^1,u_i^2)$ share relayed key $k_i=(k_i^1,k_i^2)$ in the same setting as in SNC with public channel, given in Section \ref{sec:setting_SNC}.

Then by using $k_i$ thus obtained to encrypt message $m_i$ by the one-time pad (OTP) encryption scheme \cite{KatzLindell}, one obtains $L'$.
Here the OTP encryption scheme is the following: User $u_i^1$ encrypts $m_i$ as the ciphertext $c_i=m_i+k_i^1$ and sends it to $u_i^2$ via public channel.
Then $u_i^2$ decrypts it as $\hat{m}_i=c_i+k_i^2$.

The soundness of $L'$ is obvious from the construction.
The secrecy of $L'$ follows from that of $L$, since the adversary's information are the same in $L$ and $L'$.
Indeed, in $L'$, the secrecy of $r_e$ on non-wiretapped edges $e$ is obvious by the construction, and thus the security of $k_i$ from that of $L$.
Then the secrecy of $m_i$ follows from that of the OTP.
This completes the proof of item 1 of Theorem \ref{thm:SNC_and_KR}.

\begin{figure}[htbp]
\begin{center}
\includegraphics[bb=0 0 900 375, width=\linewidth, clip]{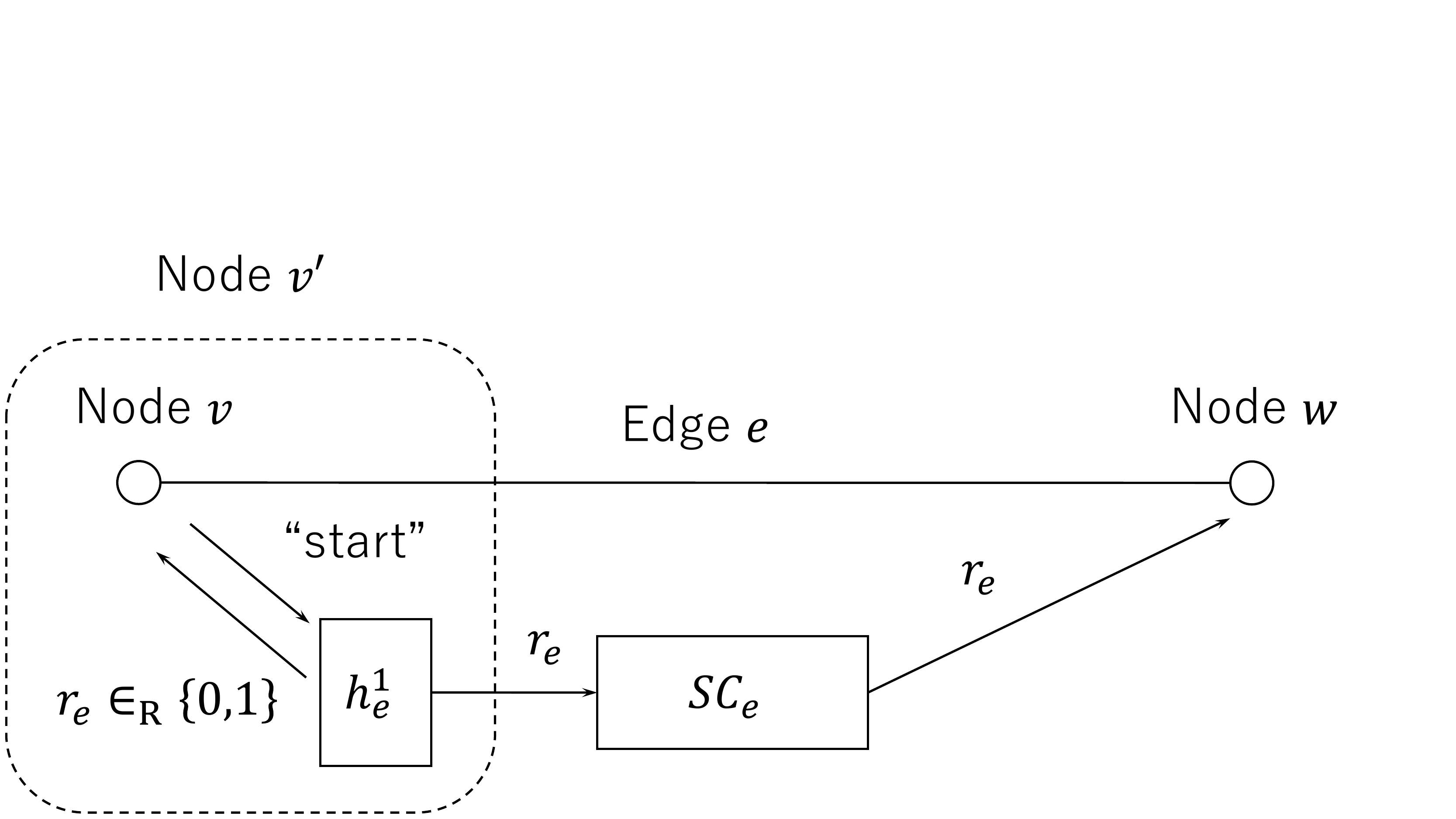}
\end{center}
 \caption{Construction for simulating a local key source $LKS_e$ (Definition \ref{dfn:random_source} and Fig. \ref{fig:random_source}) by using a secret channel $SC_e$.
We add a function $h_e^1$ to an end node $v$ of $e$ (the one that would start $LKS_e$), and regard them as a new node $v'$.
Function $h_e^2$ operates as follows: When it receives ``start'' command from $v$, it generates a uniformly random bit $r_e\in_{\rm R} \{0,1 \}$ and sends it to $SC_e$. 
}
 \label{fig:equivalence_model_kr}
\end{figure}

For the proof of item 2), note that $SC_e$ can be simulated by the local key source $LKS_e$ and the public channel $PC_e$: 
When an end node $u$ wishes to send a bit $s_e$ to the other end node $v$, it first distributes a random bit $r_e$ by switching on the local key source $LKS_e$.
Then $u$ sends $s_e$ to $v$ secretly by encrypting it by the OTP encryption scheme with $r_e$ being the secret key (Fig. \ref{fig:equivalence_model_SNC}).

By applying this construction to all secret channels included in $L$, one obtains a new KRP, which we denote by $L'$.
By construction, it is obvious that message $m_i$ as well as the adversary's information are the same, whether in $L$ or in $L'$.
This completes the proof of item 2 of Theorem \ref{thm:SNC_and_KR}.

\begin{figure}[htbp]
\begin{center}
\includegraphics[bb=0 0 960 450, width=\linewidth, clip]{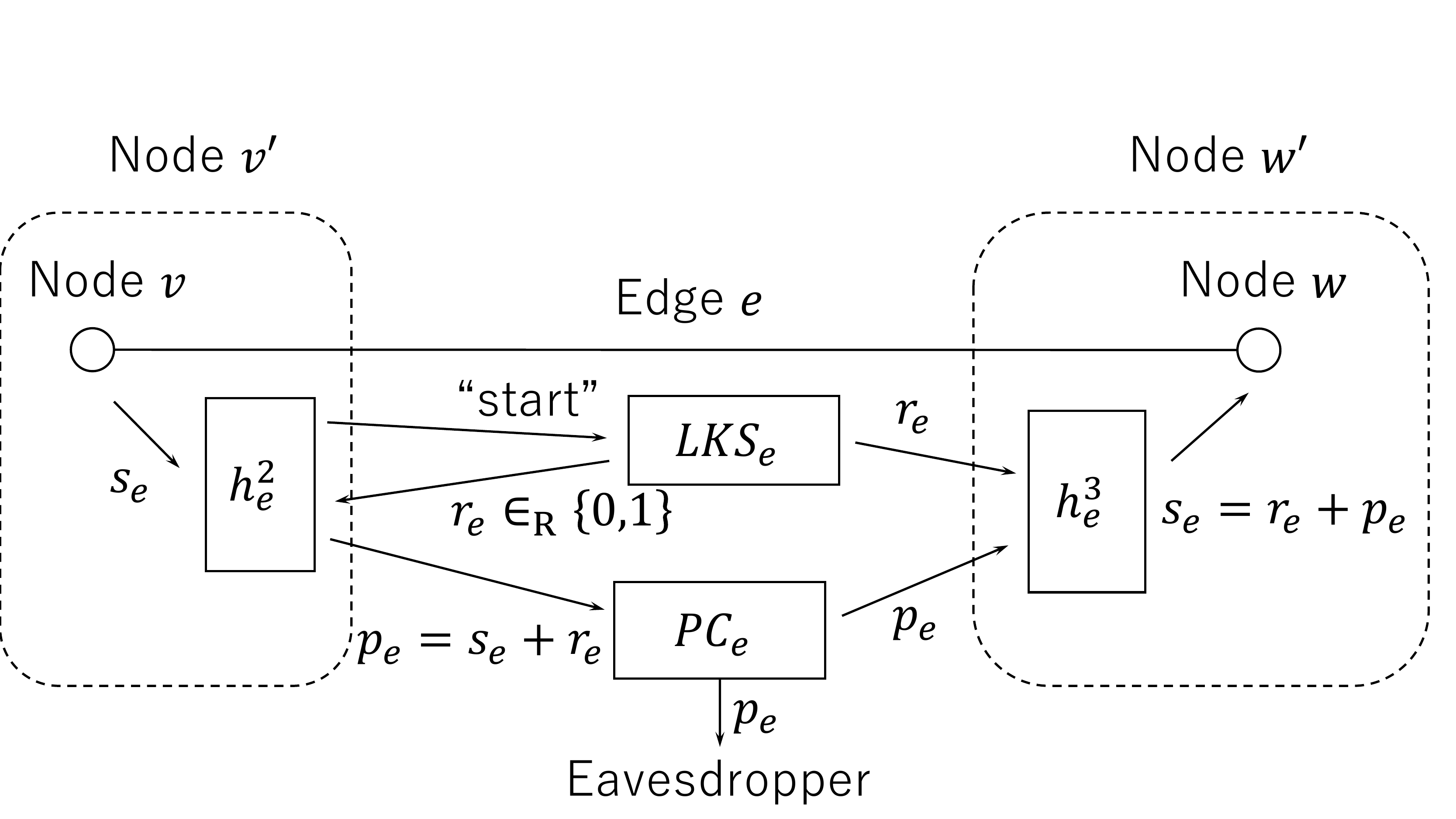}
\end{center}
 \caption{
Construction for simulating a secret channel $SC_e$ (Definition \ref{dfn:channels} and Fig. \ref{fig:channels}(a)) by using the local key source $LKS_e$ and the public channel $PC_e$.
We add a function $h_e^2$ to an end node $v$ of $e$ (the one that would start $LKS_e$), and regard them as a new node $v'$.
Function $h_e^2$ has two operations, namely, (i) on receiving $s_e$ from $u$, $h_e^2$ sends out ``start'' command to $LKS_e$, and (ii) on receiving $r_e$ from $PC_e$ $h_e^2$ sends out $p_e=s_e+r_e$ to $PC_e$.
Similarly, we add a function $h_e^3$ to the other end node $w$, and regard them as a new node $w'$.
Function $h_e^3$ has one operation: On receiving $r_e$ from $LKS_e$ and $p_e$ from $PC_e$, $h_e^3$ sends out $s_e=r_e+p_e$ to $w$.
}
 \label{fig:equivalence_model_SNC}
\end{figure}

\section{Proof of Theorem \ref{crl:SNCs_and_KRs_are_not_equivalent}}
\label{sec:Proof_of_Theorem2}
Theorem \ref{crl:SNCs_and_KRs_are_not_equivalent} asserts that the difference of structure between the KRP and the conventional SNC (shown in the first and the third columns of Table \ref{table:differences_similarity_KRs_SNCs}, and also explained in the first paragraph of Section \ref{sec:main_result})
cause a definite gap in security
(Fig. \ref{fig:inclusion_relation}).
Below, we prove this theorem by proving a set of more general lemmas (Fig. \ref{fig:inclusion_relation_2}).

That is, we introduce another new type of protocols that we call {\it KRP-by-SNC} (KRP by the setting of SNC without public channels), which corresponds to the fourth column of Table \ref{table:differences_similarity_KRs_SNCs}.
Then we show that secure schemes satisfy the relations
\begin{itemize}
\item[] Conventional SNC $\subseteq$ KRP-by-SNC $\subseteq$ KRP
\end{itemize}
(Lemma \ref{thm:SNC_var} and \ref{thm:not_more_secure_than_KRP}), as well as
\begin{itemize}
\item[] KRP-by-SNC $\ne$ KRP.
\end{itemize}
(Lemma \ref{thm:more_secure_than_SNCs}).
These two relations together assert that the KRP is strictly more secure than the conventional SNC, which completes the proof of Theorem \ref{crl:SNCs_and_KRs_are_not_equivalent}.





\begin{figure}[htbp]
\centering
\includegraphics[scale=0.3]{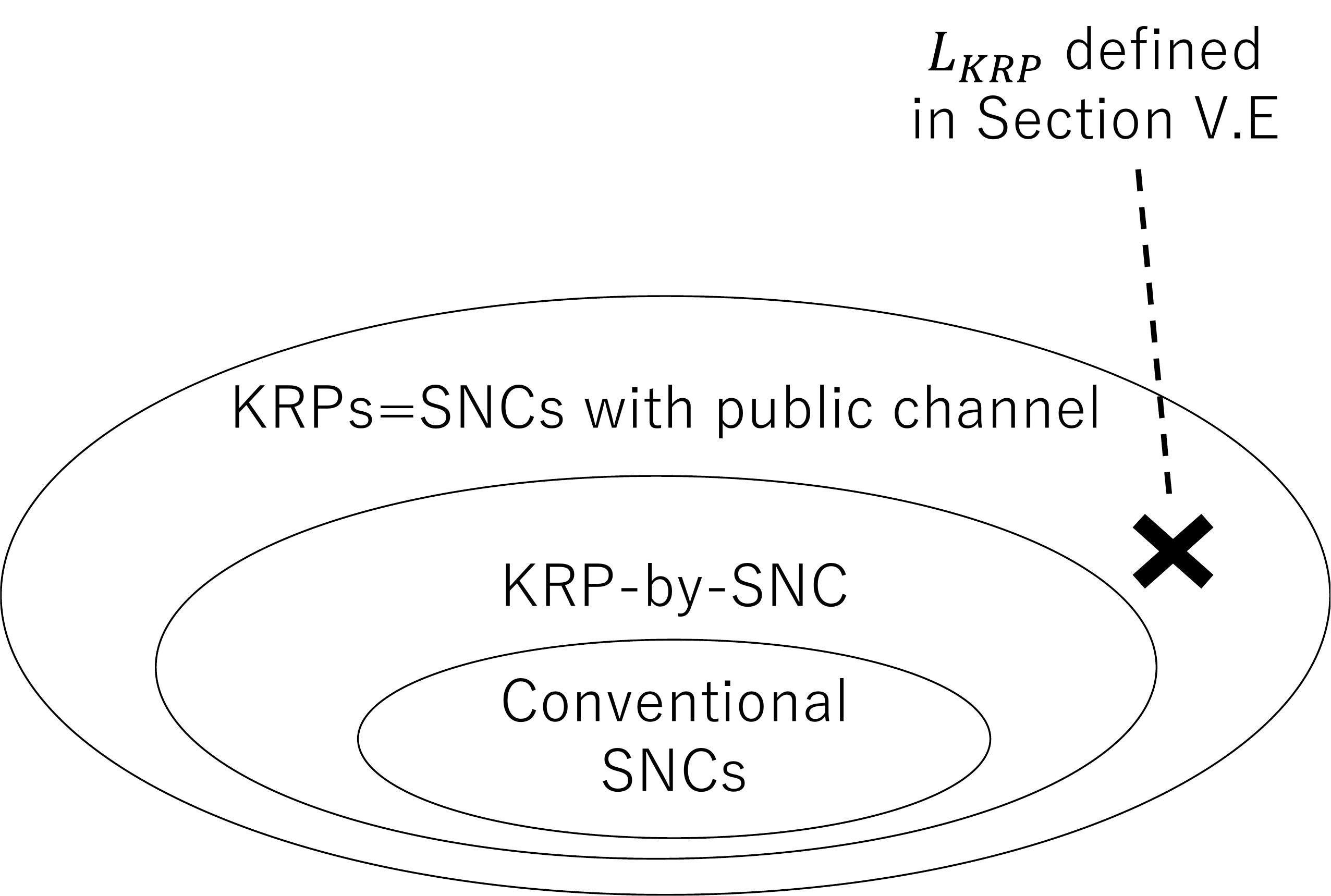}
 \caption{
Inclusion relation of secure network coding (SNC) schemes with and without public channels, key relay protocols (KRPs), and key relay protocols by using SNC without public channels (KRPs-by-SNC).
}
 \label{fig:inclusion_relation_2}
\end{figure}

\subsection{Definition of KRP-by-SNC}

KRP-by-SNC (KRP by the setting of SNC without public channels) is a variant of SNC corresponding to the fourth column of Table \ref{table:differences_similarity_KRs_SNCs}.
That is, it uses the same setting as in the conventional SNC, where players $V$ use the secret channels $SC_e$ only.
On the other hand, the goal is the same as in the KRP, where user pairs $u_i$ aim to share a random bit $k_i$.

Alternatively, KRP-by-SNC can also be considered as a variant of KRP, obtained by replacing  the public channels $PC_e$ and the local key sources $LKS_e$ with the secret channels $SC_e$.



The formal definition of KRP-by-SNC is as follows.
\subsubsection{Setting}
\label{sec:setting_SNC_kg}
The setting is the same as that of SNC without public channels (the conventional SNC), except that
\begin{itemize}
\item We denote user pairs by $u_i=(u_i^1,u_i^2)$, as in the case of KRP.
We make no distinction between the sender and the receiver.
\end{itemize}
We use this notation to stress that our goal here is to distribute a random bit $k_i$, and thus no particular user ($u_i^1$ or $u_i^2$) is entitled to choose the value $k_i$.

\subsubsection{KRP-by-SNC}
\label{sec:security_critetria_KRP-by-SNC}
The goal here is the same as that of KRP: Each user pair $u_i=(u^1_i,u^2_i)$ share a uniformly random key $k_i=(k_i^1,k_i^2)$, without disseminating it to the adversary.
Thus the basic form of the protocol should be the same as that of the KRP given in Definition \ref{dfn:key_relay}.
However, as the setting here is different ($SC_e$ are used instead of $PC_e$ and $LKS_e$), we need to modify Definition \ref{dfn:key_relay} as follows.
\begin{Dfn}[KRP-by-SNC]
\label{dfn:SNC_kg}
A protocol $L$ of the following type, performed by players $V$, is called a KRP-by-SNC.
\begin{itemize}
\item Players $V$ communicate by using  secret channels $SC_e$\footnote{We also assume that the outputs of players are defined as functions of previously received data and of random variables generated by the player, and that each player sends out the output whenever necessary data are all received.}.

Here, each $SC_e$ can only be used once.
\item  Each user $u_i^j$ calculates its relayed key $k_i^j$.
\end{itemize}
\end{Dfn}

\subsubsection{Security criteria of KRP-by-SNC}
\label{sec:security_criteria_SNC_kg}
We use the same security criteria as in the KRP, namely, Definition \ref{dfn:security_criteria}.
However, there is a caveat here: In the present case of KRP-by-SNC, the adversary's information $A(E_l^{\rm adv})$ appearing in Definition \ref{dfn:security_criteria} consists of the secret bits $s_e$ on wiretapped edges $e\in E_l^{\rm adv}$.
This is because, in KRP-by-SNC, players $V$ use the secret channels $SC_e$ only.

%


\subsection{Proof of Theorem \ref{crl:SNCs_and_KRs_are_not_equivalent}}

The following two lemmas assert that the security of KRP-by-SNC is between those of SNC and KRP.


\begin{Lmm}[Secure conventional SNC $\subseteq$ Secure KRP-by-SNC]
\label{thm:SNC_var}
KRP-by-SNCs can always achieve the same security as the conventional SNC schemes.
That is, given a conventional SNC scheme $L$ compatible with a graph $G$ and a sender-receiver configuration $u_i=(a_i,b_i)$ which is secure against wiretap sets ${\cal E}^{\rm adv}$, one can construct a  KRP-by-SNC  $L'$ compatible with the same $G$ and $u_i$ which is secure against the same ${\cal E}^{\rm adv}$.
\end{Lmm}

\begin{proof}
KRP-by-SNC $L'$ can be realized by letting the sender $a_i$ of SNC $L$ choose a random bit $k_i$ and send it out as a message $m_i$.
It is straightforward to verify that the security of $L'$, defined in Section \ref{sec:security_critetria_KRP-by-SNC}, follows from the security of $L$, defined in Section \ref{sec:security_criteria_SNC}.
\end{proof}

\begin{Lmm}[Secure KRP-by-SNC $\subseteq$ Secure KRP]
\label{thm:not_more_secure_than_KRP}
KRP can always achieve the same security as KRP-by-SNC.
That is,
given a KRP-by-SNC $L$ compatible with a graph $G$ and a user configuration $u_i$ which is secure against wiretap sets ${\cal E}^{\rm adv}$, one can construct a KRP  $L'$ compatible with the same $G$ and $u_i$ which is secure against the same ${\cal E}^{\rm adv}$.
\end{Lmm}
\begin{proof}
This can be proved in the same manner as in the proof of item 2) of Theorem \ref{thm:SNC_and_KR}.
By rewriting all the secret channels $SC_e$ appearing in KRP-by-SNC $L$ as the OTP-encrypted channels using $LKS_e$ and $PC_e$, we obtain KRP $L'$.
\end{proof}

According to Lemma \ref{thm:not_more_secure_than_KRP} above, there still remains the possibility that the securities of KRP-by-SNC and KRP are equal.
Lemma \ref{thm:more_secure_than_SNCs} below disproves this possibility.
\begin{Lmm}[Secure KRP-by-SNC $\ne$ Secure KRP]
\label{thm:more_secure_than_SNCs}
For a graph $G_0$, a user configuration $u_i$ (defined in Section \ref{sec:construction_of_counterexample}), and the empty wiretap set ${\cal E}^{\rm adv}=\{\varnothing\}$, there exists a secure KRP $L_{\rm KRP}$ (defined in Section \ref{sec:L_KRP_defined}), but there exists no secure KRP-by-SNC.
\end{Lmm}

Hence there is a definite gap in security between KRP-by-SNC and KRP.
Combined with Lemma \ref{thm:SNC_var}, this also means that there is a definite gap in security between KRP and the conventional SNC, which completes the proof of 
Theorem \ref{crl:SNCs_and_KRs_are_not_equivalent}.

The rest of this section is devoted to the proof of Lemma \ref{thm:more_secure_than_SNCs}.
The outline of the proof is as follows.
First, we define a graph $G_0$ and a configuration of user pairs $u_i$ (Section \ref{sec:construction_of_counterexample}), as well as a KRP $L_{\rm KRP}$ compatible with them (Section \ref{sec:L_KRP_defined}).
Then we show that $L_{\rm KRP}$ is secure secure against $\mathcal E^{{\rm adv}, G_0}:=\{\varnothing\}$, i.e., secure when no edge is wiretapped (Section \ref{sec:security_of_L}).
Then we show that there exists no secure KRP-by-SNC compatible with the same $G_0$ and $u_i$ even when no edge is wiretapped (Section \ref{sec:proof_of_non_existence}).

\subsection{Notation}
For ease of notation, we will often write $r_e$ and $p_e$ defined in Definition \ref{dfn:random_source} as $r[e]$ and $p[e]$.
Similarly, we often write $s_e$ defined in Definition \ref{dfn:channels} as $s[e]$.

\subsection{Description of the graph $G_0$ and the user configuration $u_i$}
\label{sec:construction_of_counterexample}

\subsubsection{Description using figures}
We define the  graph $G_0$ by a nested structure as in Figs. \ref{fig:sub-graph}, \ref{fig:G_0}, and \ref{fig:connectivity}.

That is, we first define a sub-graph $G^{\rm bn}$ by Fig. \ref{fig:sub-graph}, which is in fact the well-known modified butterfly network \cite{ho_lun_2008}.
Then we construct the graph $G_0$ by connecting subgraphs\footnote{Copies of the subgraph $G^{\rm bn}$ labeled with indices $s,i$.} $G^{\rm bn}_{s,i}$ and users $u_i^j$, as in Figs. \ref{fig:G_0} and \ref{fig:connectivity}.

\begin{figure}[htbp]
\centering
\includegraphics[scale=0.3]{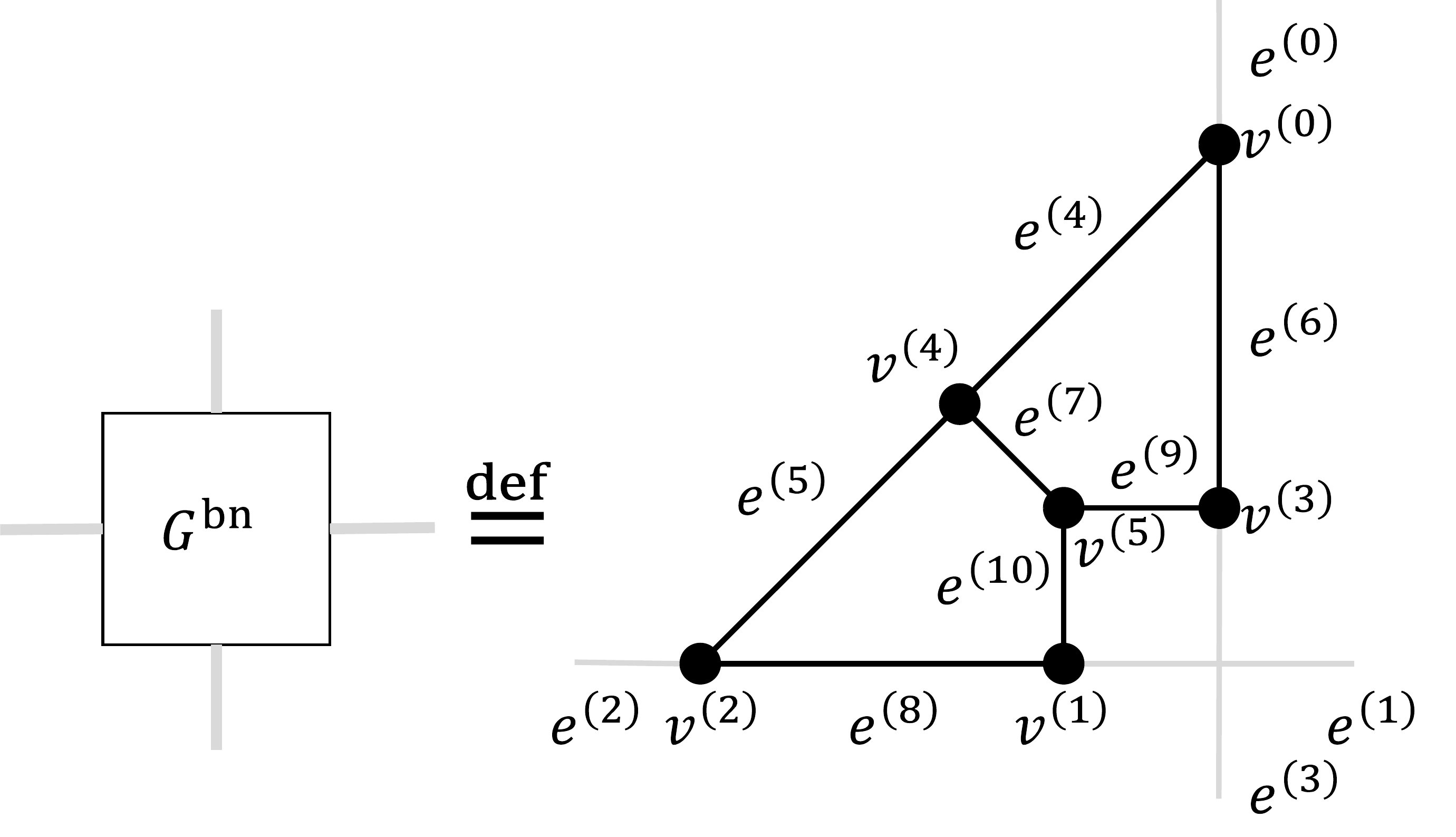}
 \caption{
Sub-graph $G^{\rm bn}$ is defined by the black nodes and edges on the right hand side.
This subgraph $G^{\rm bn}$ is in fact the well-known modified butterfly network.
When there are multiple copies of $G^{\rm bn}$, we distinguish them as $G^{\rm bn}_{s,i}$ by indices $s,i$.
Gray edges are the external edges which connect $G^{\rm bn}_{s,i}$ with another subgraph $G^{\rm bn}_{s',i'}$ or with a user $u_i^1$ or $u_i^2$.
 }
 \label{fig:sub-graph}
\end{figure}
\begin{figure}[htbp]
\centering
 \includegraphics[scale=0.5]{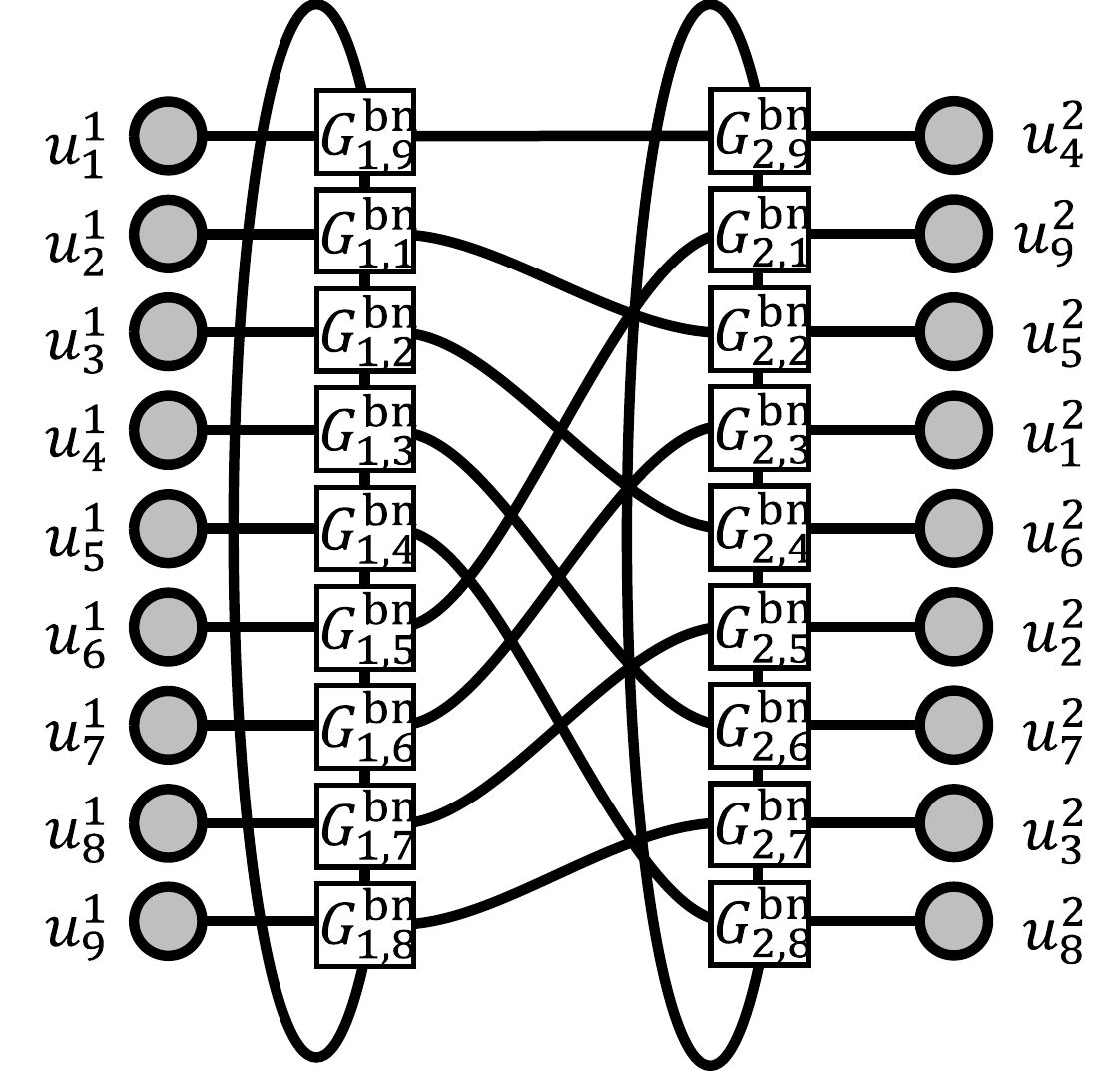}
 \caption{The graph $G_0$ and the user configuration $u_i$.
Subgraphs $G^{\rm bn}_{s,i}$ are copies of the subgraph $G^{\rm bn}$ defined in Fig. \ref{fig:sub-graph}.
Edges are wired according to the rule of Fig. \ref{fig:connectivity}.
}
 \label{fig:G_0}
\end{figure}

\begin{figure}[htbp]
\centering
\includegraphics[scale=0.3]{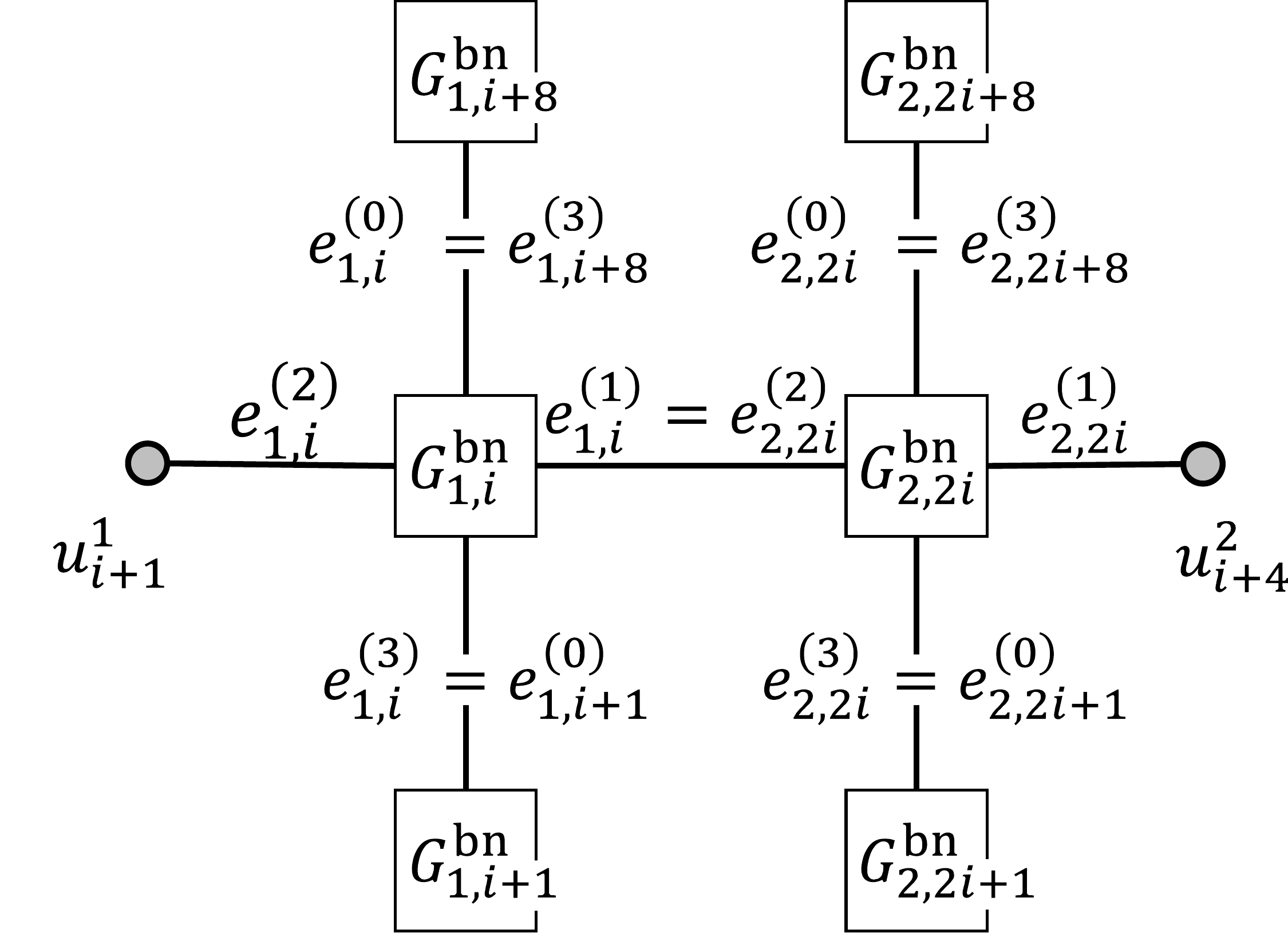}
 \caption{
The fundamental wiring rule of the graph $G_0$ depicted in Fig. \ref{fig:G_0}.
One can reconstruct graph $G_0$ by repeating this rule.
 }
\label{fig:connectivity}
 \end{figure}

\subsubsection{Alternative description using equations}

The same graph $G_0=(V_0,E_0)$ and the user pairs $u_i=(u_i^1,u_i^2)$ can also be defined using equations, as follows.

The node set $V_0$ consists of the user pairs $u_i^s$, and of the nodes $v_{s,i}^{(\alpha)}$ composing subgraphs $G^{\rm bn}_{s,i}$, where indices $s,i,\alpha$ are in the ranges $s\in\{1,2\}$, $i\in\mathbb Z/9\mathbb Z$\footnote{We use this notation to suggest that modulo 9 is implied in arithmetic involving variable $i$.} and $\alpha\in\{0,1,2,3,4,5\}$.
 
The edge set $E_0$ consists of the internal edges of subgraphs $G^{\rm bn}_{s,i}$,
\begin{align}
\{v^{(0)},v^{(4)}\}&=e^{(4)},
&\{v^{(2)},v^{(4)}\}&=e^{(5)},
&\{v^{(0)},v^{(3)}\}&=e^{(6)},
\nonumber\\
\{v^{(4)},v^{(5)}\}&=e^{(7)},
&\{v^{(2)},v^{(1)}\}&=e^{(8)},
&\{v^{(5)},v^{(3)}\}&=e^{(9)},
\nonumber\\
\{v^{(5)},v^{(1)}\}&=e^{(10)},
\end{align}
edges connecting different subgraphs $G^{\rm bn}_{s,i}$, 
\begin{align}
\{v_{1,i}^{(1)},v_{2,2i}^{(2)}\}&=e_{1,i}^{(1)}=e_{2,2i}^{(2)},\\
\{v_{s,i}^{(3)},v_{s,i+1}^{(0)}\}&=e_{s,i}^{(3)}=e_{s,i+1}^{(0)},
\end{align}
and edges connecting subgraphs $G^{\rm bn}_{s,i}$ and users $u_i^s$:
\begin{align}
\{u_{i+1}^{1},v_{1,i}^{(2)}\}=e_{1,i}^{(2)},\ \{v_{2,2i}^{(1)},u_{i+4}^{2}\}&=e_{2,2i}^{(1)},
\end{align}
where $s\in\{1,2\}$, $i\in \mathbb Z/9\mathbb Z$.



\subsubsection{Notation related with  $G^{\rm bn}_{s,i}$}
Below, for ease of notation, we will often suppress subscripts $s,i$ (corresponding to one of subgraphs $G^{\rm bn}_{s,i}$), if it is clear from the context which subgraph $G^{\rm bn}_{s,i}$ we focus on.
Also, whenever we say a subgraph $G_{s,i}^{\rm bn}$ is a sender/receiver, it means that a node inside $G_{s,i}^{\rm bn}$ is a sender/receiver.

\subsection{KRP $L_{\rm KRP}$ compatible with $G_0$ and $u_i$}

\subsubsection{Construction of $L_{\rm KRP}$}
\label{sec:L_KRP_defined}
%
%
%


\paragraph{First step}
\label{sec:L_KRP_defined_first_step}
%
The goal here is that: For all $s,i$ and $\beta\in\{0,2\}$, the node $v_{s,i}^{(\beta)}$ sends  the local key $r[e_{s,i}^{(\beta)}]$ to the node 
$v_{s,i}^{(\beta+1)}$ secretly.

%
In order to realize this task, we use the following idea:
Note that, if secret channels $SC_e$ were available, this task could be realized by using the well-known modified butterfly network coding in each sub-graph $G_{s,i}^{\rm bn}$ \cite{ho_lun_2008}.
However, since we do not have secret channels $SC_e$ in the present setting, we emulate them with the one-time pad (OTP) encrpytion using local keys supplied by $LKS_{e}$ and with public communication on $PC_e$ (as we did in the proof of Theorem \ref{thm:SNC_and_KR}; also see Fig. \ref{fig:equivalence_model_SNC}).
Namely, whenever a node $v^{(\alpha)}$ wishes to secretly transmit a bit $r$ to an adjecent node $v^{(\alpha')}$, it encrypts $r$ using the key supplied by $LKS_{e}$ on the edge $e$ between $v^{(\alpha)}$ and $v^{(\alpha')}$, and sends it to $v^{(\alpha')}$ via the public channel $PC_e$.
In the following, we often write this emulated secret transmission as $v^{(\alpha)}\rightarrow v^{(\alpha')}:r$.

Due to this idea and notation, our first step of $L_{\rm KRP}$ takes the following form (Fig. \ref{fig:BNC}):
\begin{align}
v^{(0)}\rightarrow v^{(3)}:&r[e^{(0)}]
\\
v^{(0)}\rightarrow v^{(4)}:&r[e^{(0)}]
\\
v^{(2)}\rightarrow v^{(1)}:&r[e^{(2)}]
\\
v^{(2)}\rightarrow v^{(4)}:&r[e^{(2)}]
\\
v^{(4)}\rightarrow v^{(5)}:&r[e^{(0)}]\oplus  r[e^{(2)}]
\\
v^{(5)}\rightarrow v^{(1)}:&r[e^{(0)}]\oplus  r[e^{(2)}]
\\
v^{(5)}\rightarrow v^{(3)}:&r[e^{(0)}]\oplus  r[e^{(2)}],
\end{align}
and $v^{(1)}$ and $v^{(3)}$ obtain the values $(r[e^{(0)}]\oplus  r[e^{(2)}])\oplus r[e^{(2)}]$ and $(r[e^{(0)}]\oplus  r[e^{(2)}])\oplus  r[e^{(0)}]$
respectively.

\begin{figure}[htbp]
\centering
  \includegraphics[scale=0.3]{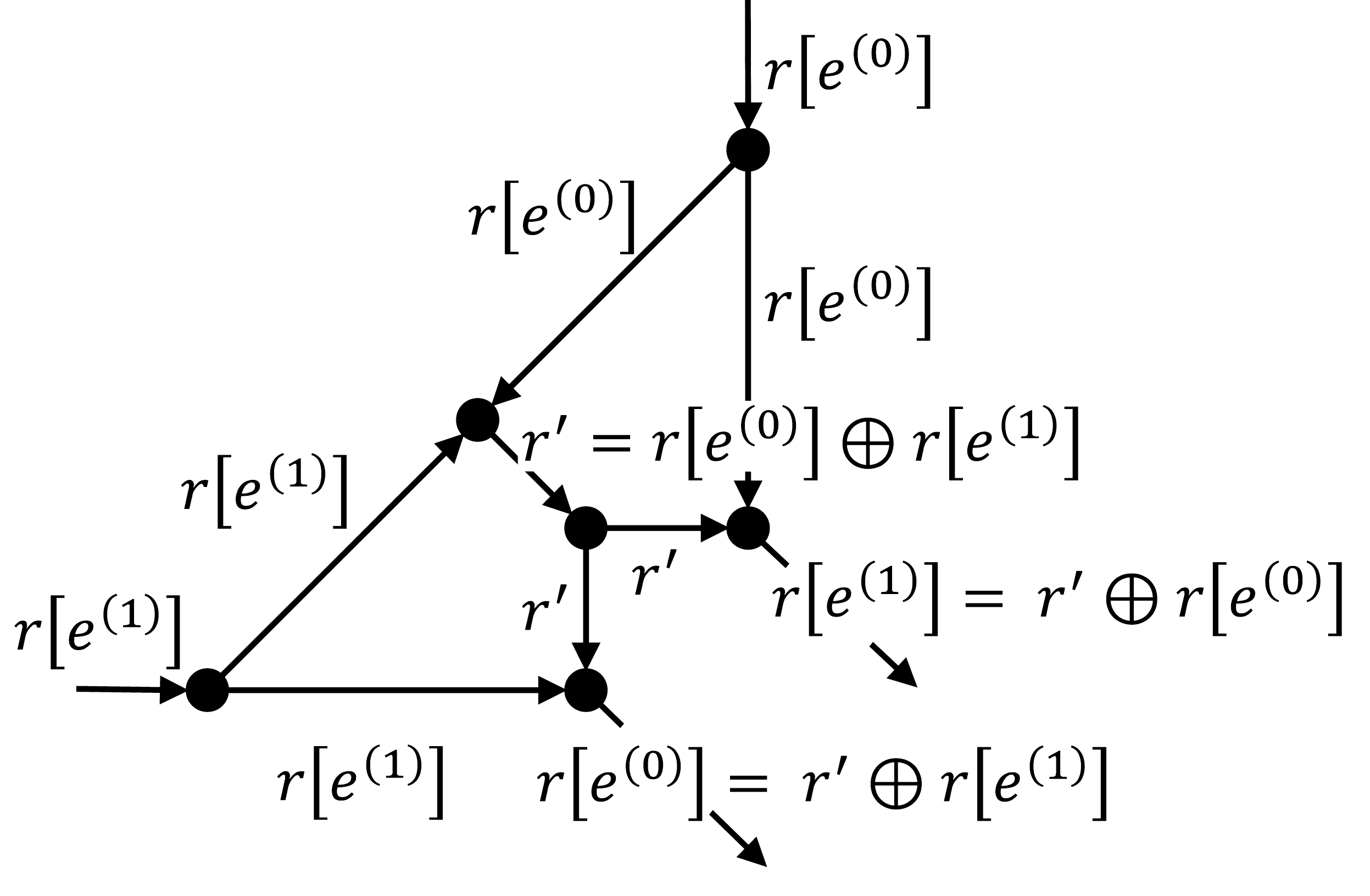}
 \caption{
The first step of $L_{\rm KRP}$, given in Section \ref{sec:L_KRP_defined_first_step}.
This is essentially the same as the well-known modified butterfly network coding.
}
 \label{fig:BNC}
\end{figure}

\paragraph{Second step}
Next, using the local keys $r[e_{s,i}^{(\beta)}]$ thus obtained, each user pair $u_i=(u_i^1,u_i^2)$ shares a relayed key.
They do this by performing the serial KRP of Fig. \ref{fig:key_relay_example2} (a), on a straight path consisting of  edges
  $v_{1,i+8}^{(3)}$, 
  $v_{1,i}^{(1)}$, 
  $v_{2,2i}^{(3)}$ 
  and $v_{2,2i+1}^{(1)}$.

Namely,  the nodes in the middle of the path announce 
$r[e^{(2)}_{1,i+8}]\oplus r[e^{(0)}_{1,i}]$,
%
$r[e^{(0)}_{1,i}]\oplus r[e^{(2)}_{2,2i}]$,
%
$r[e^{(2)}_{2,2i}]\oplus r[e^{(0)}_{2,2i+1}]$ and
%
$r[e^{(0)}_{2,2i+1}]\oplus r[e^{(1)}_{2,2i+1}]$, and then
the user $u_i^2$ obtains the bit $r[e^{(2)}_{1,i+8}]$
by summing up the published bits and the local key $r[e^{(1)}_{2,2i+1}]$.

\subsubsection{Security of $L_{\rm KRP}$}
\label{sec:security_of_L}

From the construction above, we immediately have the following lemma.

\begin{Lmm}[Security of $L_{\rm KRP}$]
\label{lmm:security_L}
The  KRP $L_{\rm KRP}$ is secure against wiretap sets ${\cal E}^{{{\rm adv},G_0}}=\{\varnothing\}$. 
\end{Lmm}

\begin{proof}
We give a detailed proof in Appendix.
The basic idea of the proof is as follows.

The security (the secrecy and the soundness) of the emulated secret channels used in the first step is obvious by construction, and thus the security of the local key $r[e_{s,i}^{(\beta)}]$ 
is guaranteed by that of the modified butterfly network coding.
The security of the serial KRP is also obvious by construction.
These two facts together guarantee the security of  $L_{\rm KRP}$.
\end{proof}

\subsection{Proof that there exists no secure KRP-by-SNC compatible with $G_0$ and $u_i$}
\label{sec:proof_of_non_existence}
Lemma \ref{thm:more_secure_than_SNCs} asserts that there exists no secure KRP-by-SNC compatible with $G_0$ and $u_i$ defined above, even when no edge is wiretapped.
Below we prove this assertion by first supposing that such KRP-by-SNC exists and then deriving a contradiction.

In order to describe the contradiction, it is convenient to number edges $e$ according to the time when the secret channel $SC_e$ is used.
Obviously, such numbering is possible for any KRP-by-SNC.
Formally, this numbering is equivalent to the following total order $\prec$.
\begin{Dfn}[Total order $\prec$ of edges]
\label{def:order_edge_use}
Given a KRP-by-SNC $L$ on a graph $G=(V,E)$, we write $e\prec e'$ ($e\in E$ is smaller than $e'\in E$) if  the secret channel $SC_e$ is used before $SC_{e'}$ is used, in $L$.
\end{Dfn}

Now, if we suppose that there exists a secure KRP-by-SNC $L_0$ (compatible with $G_0$, $u_i$, and ${\cal E}^{{\rm adv}}=\{\varnothing\}$), $L_0$ must satisfy the condition of the following lemma.
\begin{Lmm}
\label{lmm:cg_graph_property}

If a KRP-by-SNC $L_0$ compatible with $G_0$ and $u_i$ is secure against ${\cal E}^{\rm adv}=\{\varnothing\}$, 
then at least one of subgraphs $G_{s,i}^{\rm bn}\subset G_0$ must satisfy the following four requirements:
\begin{enumerate}[\rm R1]
\item For $\beta\in\{0,2\}$ each, the secret bits conveyed on the edges  $e_{s,i}^{(\beta)}$ and  $e_{s,i}^{(\beta+1)}$ are completely random and completely correlated,  i.e. $I(S[e_{s,i}^{(\beta)}]:S[e_{s,i}^{(\beta+1)}])=1$.
\item The secret bit on $e_{s,i}^{(0)}$ is independent of that on $e_{s,i}^{(2)}$, i.e. $I(S[e_{s,i}^{(0)}]:S[e_{s,i}^{(2)}])=0$. 
\item 
For $\beta\in\{0,2\}$ each, $G_{s,i}^{\rm bn}$ is the sender of the larger edge of  $e_{s,i}^{(\beta)}$
 and $e_{s,i}^{(\beta+1)}$.
 \item $G_{s,i}^{\rm bn}$ is the receiver of the second largest edge in the set $\{e_{s,i}^{(\alpha)}\}_{\alpha\in\{0,1,2,3\}}$.
\end{enumerate}
 \end{Lmm}

We will prove this lemma in Section \ref{sec:KRP_by_SNC_satisfies_Ri}.




However, we can also show that no KRP-by-SNC can satisfy such condition:
\begin{Lmm}
\label{lmm:property_butterfly_network}
In any KRP-by-SNC $L_0$ compatible with $G_0$ and $u_i$, no subgraph $G_{s,i}^{\rm bn}$ can satisfy the four requirements R1, $\dots$, R4 of Lemma \ref{lmm:cg_graph_property} simultaneously.
\end{Lmm}
This is a contradiction, and thus a secure $L_0$ cannot exist. This completes the proof of Lemma \ref{thm:more_secure_than_SNCs}.

\begin{proof}[Proof of Lemma \ref{lmm:property_butterfly_network}]




We suppose that one of subgraphs $G_{s,i}^{\rm bn}$ satisfies R1, $\dots$, R4, and derive a contradiction.

Below, as we concentrate on one such $G_{s,i}^{\rm bn}$ satisfying R1, $\dots$, R4, we omit subscripts $s,i$ for ease of notation, on the subgraph $G_{s,i}^{\rm bn}$, edges $e_{s,i}^{(\alpha)}$, and nodes $v_{s,i}^{(\alpha)}$.
Also, as KRP-by-SNC here uses only one type of channels, $SC_e$, we will often identify an edge $e$ with the secret channel $SC_e$.


First note that the pair of the two largest edges in $\{e^{(\alpha)}\}_{\alpha\in\{0,1,2,3\}}$ must either be $e^{(0)}$ and $e^{(1)}$, or be $e^{(2)}$ and $e^{(3)}$. 
This is because if it is not the case,
then R3 says that $G^{\rm bn}$($=G_{s,i}^{\rm bn}$) is the sender of the second largest edge, but this contradicts R4. 

Then, of all the remaining cases, we will below focus on one particular case where edges $e^{(\alpha)}$ are ordered as
\begin{itemize}
\item[1)]  $e^{(0)}\prec e^{(1)}\prec  e^{(2)}\prec e^{(3)}$,
\end{itemize}
and derive a contradiction.
 (Other cases, such as $e^{(2)}\prec  e^{(3)}\prec  e^{(1)}\prec e^{(2)}$, can also be shown similarly.) 
In this case, due to R1 and R2 we have
\begin{itemize}
\item[2)] $I(S[e^{(0)}];S[e^{(1)}])=1$,
\item[3)] $I(S[e^{(2)}];S[e^{(3)}])=1$,
\item[4)] $I(S[e^{(0)}];S[e^{(2)}])=0$.
\end{itemize}
Also, $G^{\rm bn}$ must be the sender of edges $e^{(1)}$ and $e^{(3)}$ due to R3, and also the receiver of $e^{(2)}$ due to R4.
Thus we have
\begin{itemize}
\item[5)] $v^{(1)}$ is the sender of $e^{(1)}$,
\item[6)] $v^{(2)}$ is the receiver of $e^{(2)}$,
\item[7)] $v^{(3)}$ is the sender of $e^{(3)}$.
\end{itemize}
From relations 1),$\dots$,7) above, we can also prove the following two relations,
\begin{itemize}
\item[8)] 
A series of edges connecting $v^{(0)}$ and $v^{(1)}$, e.g. $e^{(4)}$, $e^{(7)}$ and $e^{(10)}$, are smaller than  $e^{(1)}$,
\item[9)]
A series of edges connecting $v^{(2)}$ and $v^{(3)}$, e.g. $e^{(5)}$, $e^{(7)}$ and $e^{(9)}$, are larger than $e^{(2)}$ and smaller than $e^{(3)}$.
 \end{itemize}
However, these two relations contradict each other, and thus Lemma \ref{lmm:property_butterfly_network} is be proved.

Item 8) is obtained as follows:
In order to realize relations  1), 2), and 5), the information of the secret bit on $e^{(0)}$ must be transferred from $v^{(0)}$ to $v^{(1)}$ before using the edge $e^{(1)}$.
For this transfer, a series of edges connecting $v^{(0)}$ and $v^{(1)}$ must be used.
Here, we have used the fact that there is no correlation between nodes before executing protocol $L_0$.

Item 9) is obtained as follows:
Relations 2) and 4) imply that $I(S[e^{(0)}],S[e^{(1)}];S[e^{(2)}])=0$.
Combining this fact and relations 1) and 6), we find that, before the secret channel on $e^{(2)}$ is used, any random variable obtained on the nodes $\{v^{(\alpha)}\}_\alpha$ is independent of $S[e^{(2)}]$.
As a result, relations 1), 3) and 7) imply that, the secret bit on $e^{(2)}$ must be transferred from $v^{(2)}$ to $v^{(3)}$ before $e^{(3)}$ is used (and of course after $e^{(2)}$ is used). 
\end{proof}

\subsection{Proof of Lemma $\ref{lmm:cg_graph_property}$}
\label{sec:KRP_by_SNC_satisfies_Ri}

We first introduce the notion of the {\it standard path} along with three lemmas related with it, and then use them to prove Lemma \ref{lmm:cg_graph_property}.

\subsubsection{Standard path and the related lemmas}
\label{sec:alias}

\begin{Dfn}[Standard path $e_i^{\left<\gamma\right>}$]
For each user pair $u_i=(u_i^1,u_i^2)$, we define {\it  standard path} $i$ connecting them via subgraphs $G^{\rm bn}$ (more precisely,  subgraphs $G_{1,i+8}^{\rm bn}$,  $G_{1,i}^{\rm bn}$,  $G_{2,2i}^{\rm bn}$, and $G_{2,2i+1}^{\rm bn}$) as in Fig. \ref{fig:connectivity_cg}.
The standard path $i$ consists of the following five edges,
\begin{align}
e_{i}^{\left<0\right>}&:=e_{1,i+8}^{(2)},                    \nonumber\\
e_{i}^{\left<1\right>}&:=e_{1,i+8}^{(3)}=e_{1,i}^{(0)},     \nonumber\\
e_{i}^{\left<2\right>}&:=e_{1,i}^{(1)}=e_{2,2i}^{(2)},       \nonumber\\
e_{i}^{\left<3\right>}&:=e_{2,2i}^{(3)}=e_{2,2i+1}^{(0)},  \nonumber\\
e_{i}^{\left<4\right>}&:=e_{2,2i+1}^{(1)}.
\end{align}
\end{Dfn}

\begin{figure}[htbp]
\centering
\includegraphics[scale=0.3]{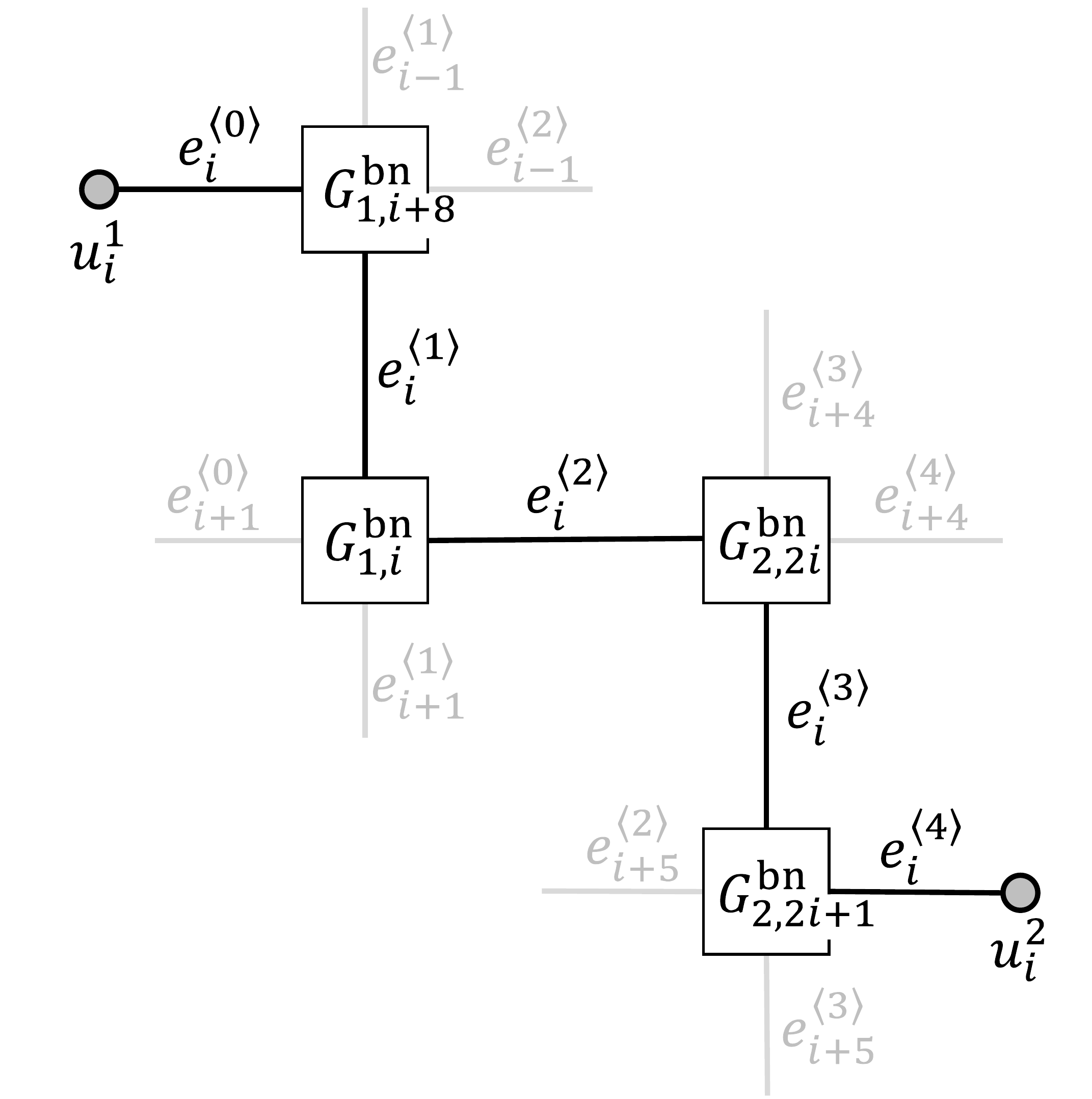}
 \caption{
The standard path $i$ ($\in\{1,\dots, 9\}$) consists of the five edges $\{e_i^{\left<\gamma\right>}\}_{\gamma\in\{0,1,2,3,4\}}$, shown in black above; also see Section \ref{sec:alias}.
Note that each standard path $i$ connects between user pair $u_i=(u_i^1,u_i^2)$.
Note also that the graph $G_0$ of Fig. \ref{fig:G_0} is a disjoint union of these paths and subgraphs $G_{s,i}^{\rm bn}$ (see Fig. \ref{fig:sub-graph}), and that every subgraph $G_{s,i}^{\rm bn}$ are connected to exactly two of these paths.
As we will see in Lemma \ref{lm:cons_chan}, in a secure KRP-by-SNC $L_0$, all the secret bits $S[e_i^{\langle \gamma\rangle}]$ transferred on the standard path $i$ must equal the relayed key $k_i^j$, up to constants.
}
 \label{fig:connectivity_cg}
\end{figure}


We can show that all the edges $e^{\langle \gamma\rangle}_i$ on a standard path $i$ conveys essentially the same information, the relayed key $k_i$.
\begin{Lmm}
\label{lm:cons_chan}
In each standard path $i$,
the secret bits $S[e_i^{\left<\gamma\right>}]$ conveyed there must be equal to the relayed key  
 $k_i^1=k_i^2$ shared by the user pair $u_i^1$ and $u_i^2$ at the end points, up to constants.
That is, 
for any $i$, $j$ and $\gamma$,
\begin{align}
S[e_i^{\left<\gamma\right>}]&=K_i^j\oplus d[e_i^{\left<\gamma\right>}]
\label{eq:main_result}
\end{align}
with $d[e_i^{\left<\gamma\right>}]$ 
being constants.
\end{Lmm}
We will prove this lemma in Section \ref{sec:proof_variables_on_SP_are_equal_up_to_constant}.

Further, we can also show that $k_i$ is first generated locally by one entity (a subgraph $G^{\rm bn}_{i',s'}$ or a user $u_i^j$) on the standard path $i$, and then repeatedly conveyed to an adjacent entity, until it is shared by the users $u_i^1$ and $u_i^2$ at the end points; see Fig. \ref{fig:Flow}.
This phenomenon can be stated formally in terms of the ordering $\prec$, as follows.

\begin{Lmm}
\label{lm:cons_ord}
 In each standard path $i$,
the edges are used  in the following 
manner.
Let $e_i^{\left<\gamma_i\right>}$ denote  the first edge used.
Then,
\begin{itemize}
\item
Edges to the upper left of $e_i^{\left<\gamma_i\right>}$
 are used in order from lower right to upper left:
\begin{align}
e_i^{\left<\gamma_i\right>}\prec
e_i^{\left<\gamma_i-1\right>}
\prec
\cdots
\prec
e_i^{\left<0\right>},
\end{align}
and secret bits $S[e_i^{\left<\gamma\right>}]$ on these edges besides  $e_i^{\left<\gamma_i\right>}$  flow left or upward.
\item
Similarly,  edges to the  lower right of $e_i^{\left<\gamma_i\right>}$
 are used in order from upper left to lower right:
\begin{align}
e_i^{\left<\gamma_i\right>}\prec
e_i^{\left<\gamma_i+1\right>}
\prec
\cdots
\prec
e_i^{\left<4\right>}.
\end{align}
and secret bits $S[e_i^{\left<\gamma\right>}]$ on these edges besides  $e_i^{\left<\gamma_i\right>}$  flow right or downward.
\end{itemize}
 \end{Lmm}
We will prove this lemma in Section \ref{sec:order_on_standard_path}.

\begin{figure}[htbp]
\centering
\includegraphics[scale=0.3]{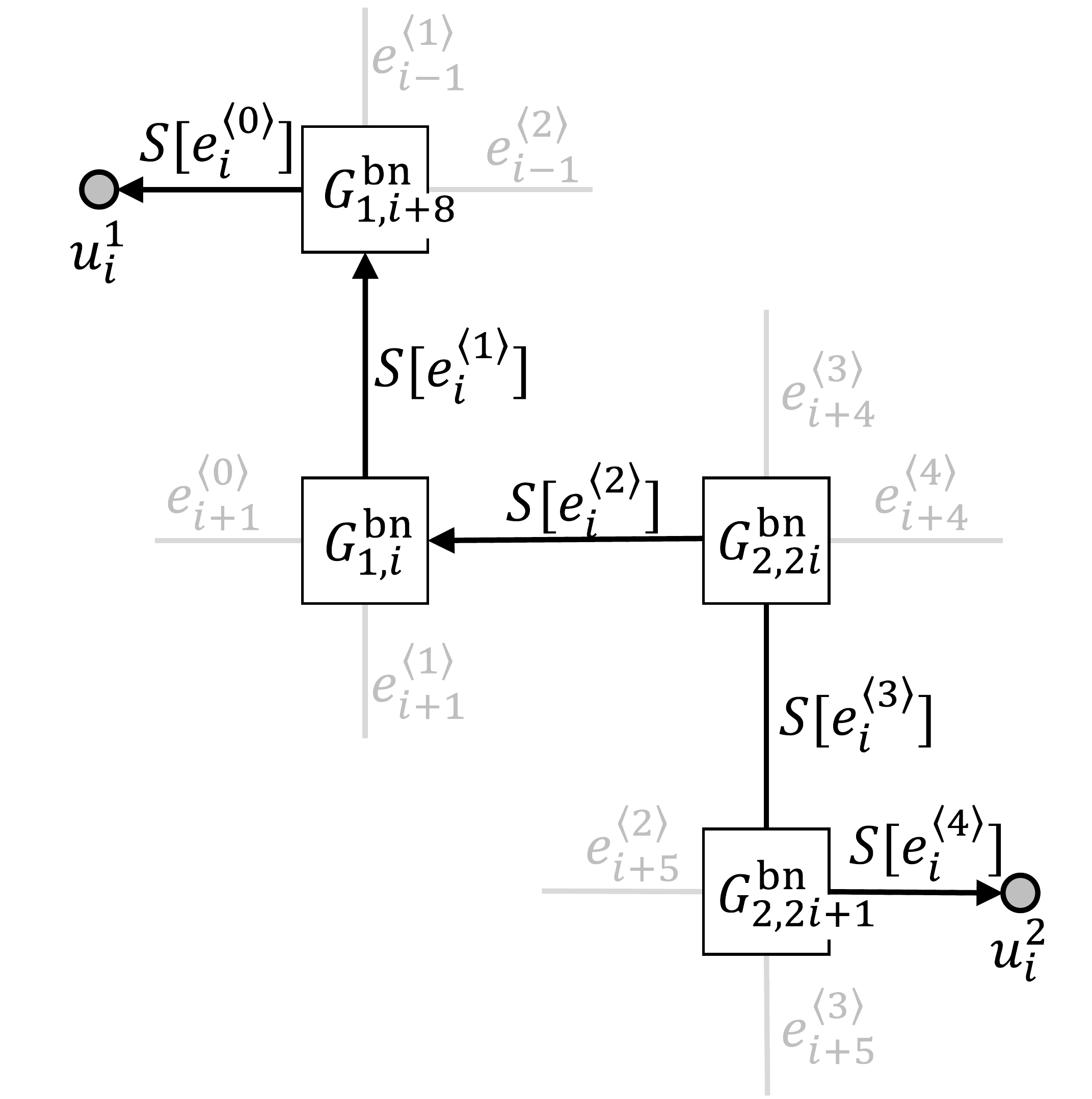}
 \caption{
An example of the flow of secret information $S[e_i^{\left<\beta\right>}]$ on a standard path $i$, as stated in Lemma \ref{lm:cons_ord}.
Recall that all the secret bits $S[e_i^{\langle \beta\rangle}]$ transferred on the standard path $i$ must equal the relayed key $k_i^j$, up to constants (Lemma \ref{lm:cons_chan} and Fig. \ref{fig:connectivity_cg}).
In the example above, $k_i^j$ is first generated by $G^{\rm bn}_{2,2i}$ or by $G^{\rm bn}_{2,2i+1}$, and then transferred via $e_i^{\left<3\right>}$.
It is then repeatedly conveyed to an adjacent entity, until it is shared by the users $u_i^1$ and $u_i^2$ at the end points.
}
 \label{fig:Flow}
\end{figure}

We can also show the following lemma.
\begin{Lmm}
\label{lm:cons_cont}
At least one of the following two conditions is false:
\begin{enumerate}[\rm C1]
\item
For any subgraph $G_{s,i}^{\rm bn}$, 
 $e_{s,i}^{(0)}$ or $e_{s,i}^{(1)}$ is the smallest edge in the standard path containing the two edges $e_{s,i}^{(0)}$ and $e_{s,i}^{(1)}$,
 when they are the two largest edges in $\{e_{s,i}^{(\alpha)}\}_{\alpha\in\{0,1,2,3\}}$.
\item
For any subgraph $G_{s,i}^{\rm bn}$, 
 $e_{s,i}^{(2)}$ or $e_{s,i}^{(3)}$ is the smallest edge in the standard path containing the two edges $e_{s,i}^{(2)}$ and $e_{s,i}^{(3)}$,
 when they are the two largest edges in $\{e_{s,i}^{(\alpha)}\}_{\alpha\in\{0,1,2,3\}}$.
\end{enumerate}
\end{Lmm}
We will prove this lemma in Section \ref{sec:proof_of_lemma_at_least_two_condition}.

\subsubsection{Proof of Lemma \ref{lmm:cg_graph_property}}

From Lemma \ref{lm:cons_chan} (respectively, from Lemma \ref{lm:cons_ord}), all the  subgraph $G_{s,i}^{\rm bn}$ satisfy 
R1 and R2 (respectively, R3) in Lemma \ref{lmm:cg_graph_property}.
Here, in order to derive R2, we used the mutual independence of $\{K_i^j\}_i$ additionally.

Hence, it remains to show that R4 holds for some $G_{s,i}^{\rm bn}$.
To this end, we suppose that R4 does not hold for any $G_{s,i}^{\rm bn}$, and derive a contradiction with Lemma \ref{lm:cons_cont}. 

This is equivalent to showing that, in any subgraph $G_{s,i}^{\rm bn}$, 
\begin{itemize}
\item[(i)] C1 and C2 hold if  R4 does not hold.
\end{itemize}
Below we will choose one of $G_{s,i}^{\rm bn}$ and show how to prove item (i).
As $G_{s,i}^{\rm bn}$ is fixed, we omit subscripts $s,i$ on the notation $e_{s,i}^{(\alpha)}$ for simplicity.

For the sake of the simplicity, we consider only the case of $e^{(0)}\prec e^{(1)}$ and $e^{(2)}\prec e^{(3)}$;
all other cases can be shown similarly.

Then, we divide the remaining situation into three cases:

First, if
 the second largest edge in the set $\{e^{(\alpha)}\}_{\alpha\in\{0,1,2,3\}}$ is $e^{(1)}$ or $e^{(3)}$,
the order $e^{(0)},e^{(2)}\prec e^{(1)},e^{(3)}$ must be derived from the assumptions $e^{(0)}\prec e^{(1)}$ and $e^{(2)}\prec e^{(3)}$, i.e. item (i)
 holds.

Second, if
the second largest edge in the set $\{e^{(\alpha)}\}_\alpha$ is $e^{(0)}$,
we find that $e^{(2)},e^{(3)}\prec e^{(0)}\prec e^{(1)}$ holds (as a result, C2 holds) and the sender of  $e^{(0)}$ must be $G_{s,i}^{\rm bn}$. 
Here, we have used the  assumption that 
R4 
does not hold.
From this fact and   $e^{(0)}\prec e^{(1)}$,
 Lemma \ref{lm:cons_ord} guarantees  that $e^{(0)}$ is the smallest edge
 in the standard path where it belongs, i.e. C1 holds.
Thus, the item (i) also holds in this case.

Finally, if
 the second largest edge in the set $\{e^{(\alpha)}\}_\alpha$ is $e^{(2)}$,
we can also show that  item (i) holds in the same way as in the previous case.

This completes the  proof of Lemma \ref{lmm:cg_graph_property}. 

\subsection{Note on notation: $v_{s,i}$ and $E_0^{\rm st}$}
\label{sec:coarse_grained edge set}
In the proofs of Lemmas \ref{lm:cons_chan}, \ref{lm:cons_ord} and \ref{lm:cons_cont} given below, we only need to consider communication between subgraphs $G_{s,i}^{\rm bn}$, and do not need to refer to communication occurring inside each $G_{s,i}^{\rm bn}$. Therefore, in order to simplify the presentation, we will often denote a subgraph $G_{s,i}^{\rm bn}$ as a single node $v_{s,i}$.
For example,  the edge set $\{u^1_i,v_{s,i}\}$ denotes the set $\{u^1_i,v_{s,i}^{(0)},v_{s,i}^{(1)},v_{s,i}^{(2)},v_{s,i}^{(3)},v_{s,i}^{(4)},v_{s,i}^{(5)}\}$.

Accordingly, we will also regard the graph $G_0$ as consisting of edges that connect nodes $v_{s,i}$ ($=G_{s,i}^{\rm bn}$) and $u^i$, which we denote by $E_0^{\rm st}$.
This edge set $E_0^{\rm st}$ in fact consists of edges on the standard paths $i$, defined in the previous section,
\begin{equation}
E_0^{\rm st}:=\{e_i^{\left<\gamma\right>}\}_{i\in\mathbb Z/9\mathbb Z, \gamma\in\{0,1,2,3,4\}}.
\end{equation}

\subsection{Proof of Lemma \ref{lm:cons_chan}}
\label{sec:proof_variables_on_SP_are_equal_up_to_constant}

\subsubsection{Case of $e_i^{\left<0\right>}$}
We divide nodes $V_0$ into $u_i^1$ and the others:
\begin{align}
V_i^{(1)}&:=\{
u_i^1\}\\
\bar V_i^{(1)}&:=V_0\backslash V_i^{(1)}
\end{align}
Since these two sets are connected by $e_i^{\langle0\rangle}$ only,
 there must be functions $f_i^{(j)}$ which satisfy the relation
\begin{align}
&
f_i^{(1)}\left(
S[e_i^{\left<0\right>}], \{C_v\}_{v\in V_i^{(1)}}
\right)=K_i^1
\nonumber\\
=&K_i^2
=f_i^{(2)}\left(
S[e_i^{\left<0\right>}], \{C_v\}_{v\in \bar V_i^{(1)}} 
\right).
\label{eq:main_1}
\end{align}
where $C_v$ denotes all the random variables possessed by node $v\in V_0$ 
  before executing protocol $L_0$.
The first (last) equality in (\ref{eq:main_1}) comes from the facts that 
the node $u_i^1$ ($u_i^2$)
are in  the set $V_i^{(1)}$($\bar V_i^{(1)}$).
The second  equality follows from the soundness of  $L_0$.

Recall that the two sets $V_i^{(1)}$ and $\bar V_i^{(1)}$ are connected by $e_i^{\langle0\rangle}$ only.
Also, note that from the setting of the KRP-by-SNC,
the secret bit $S[e_i^{\left<0\right>}]$ must be generated locally by the sender.
Then, we have 
\begin{align}
&I(S[e_i^{\left<0\right>}],\{C_v\}_{v\in V_i^{(1)}}
;
S[e_i^{\left<0\right>}],\{C_v\}_{v\in \bar V_i^{(1)}})
\nonumber\\
=&H(S[e_i^{\left<0\right>}]).
\label{eq:main_2}
\end{align}

Also, since $S[e_i^{\left<0\right>}]$ is one bit, 
\begin{align}
H(S[e_i^{\left<0\right>}])\leq 1.
\label{eq:main_2_1}
\end{align}

It remains to apply  the following lemma to relations (\ref{eq:main_1}), (\ref{eq:main_2}), (\ref{eq:main_1}), and 
  $H\left(K_i^j\right)=1$.
 \begin{Lmm}
\label{th:main_2}
If $B$, $C_0$, and $C_1$ are random variables, and $f_0$ and $f_1$ are  functions, satisfying
\begin{align}
f_0(B,C_0)&=f_1(B,C_1)=:A,
\label{ass:th1_1}
\\
H(A)
&\geq I(B,C_0;B,C_1)
= H(B)
\label{ass:th1_2}
\end{align}
with $A\in {\cal A}$, $B\in {\cal B}$ where $|{\cal A}|=|{\cal B}|$.
Then, there is a bijective function $g$ such that $g(A)=B$.
\end{Lmm}
Proof of this lemma is shown in Appendix.
By using  this lemma,
 we can find a bijective function $g$ such that
\begin{align}
S[e_i^{\left<0\right>}]&=g\left(K_i^j\right).
\end{align}
In other words,
we can find a certain constant values $d[e_i^{\left<0\right>}]\in\mathbb Z_2$ such that 
the relation
\begin{align}
K_i^j&=S[e_i^{\left<0\right>}]\oplus d[e_i^{\left<0\right>}]
\label{eq:main_bjf_1}
\end{align}
holds for $i\in\mathbb Z/9\mathbb Z$ and $j\in\{1,2\}$.
When we apply Lemma \ref{th:main_2} above, we have substituted 
$K_i^j$,
$S[e_i^{\left<0\right>}]$,
$\{C_v\}_{v\in V_i^{(1)}}$ and $\{C_v\}_{v\in \bar V_i^{ (1)}}$
 into $A$, $B$, $C_0$ and  $C_1$ in Lemma respectively. 

Note, that from the constraint of soundness $K_i^1=K_i^2$, we do not have to discriminate  $K_i^1$ and $K_i^2$ hereafter.
Therefore, we abbreviate the $K_i^j$ into $K_i$ below.

\subsubsection{Case of $e_i^{\left<4\right>}$}
This case can be shown completely in the same manner as the case of $e_i^{\left<0\right>}$.
We have the relation
\begin{align}
K_i&=S[e_i^{\left<4\right>}]\oplus d[e_i^{\left<4\right>}]
\label{eq:main_bjf_2}
\end{align}
 for a certain constant $d[e_i^{\left<4\right>}]$.

\subsubsection{Cases of $e_i^{\left<1\right>}$, $e_i^{\left<2\right>}$, and  $e_i^{\left<3\right>}$}
First, we use the same idea as in the case of $e_i^{\left<0\right>}$.
We define the sets  
\begin{align}
V^{(3)}_i&:=\{
 u_{i+1}^1,u_{i+2}^1,u_{i+8}^2,u_{i+4}^2,u_i^2,u_{i+5}^2,
\nonumber\\&\quad\quad\quad 
v_{1,i},v_{1,i+1},v_{2,2i+8},v_{2,2i},v_{2,2i+1},v_{2,2i+2}\}\\
V^{(4)}_i&:=\{
 u_{i+1}^1,u_{i+2}^1,u_{i+3}^1,u_{i+4}^2,u_{i}^2,u_{i+5}^2,
\nonumber\\&\quad\quad\quad 
v_{1,i},v_{1,i+1},v_{1,i+2},v_{2,2i},v_{2,2i+1},v_{2,2i+2}\}\\
V^{(5)}_j&:=\{
 u_{i+1}^1,u_{i+2}^1,u_{i+3}^1,u_{i+4}^2,u_{i}^2,u_{i+5}^2,u_{i+6}^2,
\nonumber\\&\quad\quad 
v_{1,i},v_{1,i+1},v_{1,i+2},v_{2,2i},v_{2,2i+1},v_{2,2i+2},v_{2,2i+4}\}\\
V^{(6)}_i&:=\{
u_{i+1}^1,u_{i+2}^1,u_{i+6}^1,u_{i+4}^2,u_{i}^2,u_{i+5}^2,
\nonumber\\&\quad\quad \quad 
v_{1,i},v_{1,i+1},v_{1,i+5},v_{2,2i},v_{2,2i+1},v_{2,2i+2}\}.
\end{align}
and $\bar V_i^{(\mu)}:=V_0\backslash V_i^{(\mu)}$.

Since the two sets $V_i^{(3)}$ and $\bar V_i^{(3)}$ are connected only by
$e_i^{\left<1\right>}$, $e_{i+2}^{\left<1\right>}$, 
$e_{i+1}^{\left<3\right>}$, $e_{i+8}^{\left<3\right>}$,
$e_{i+4}^{\left<2\right>}$, and $e_{i+5}^{\left<2\right>}$,
there exists a bijective function  $g^{(3,i)}$:
\begin{align}
&\nonumber\quad
\left(
S[e_i^{\left<1\right>}],S[e_{i+2}^{\left<1\right>}], 
S[e_{i+1}^{\left<3\right>}], S[e_{i+8}^{\left<3\right>}],
S[e_{i+4}^{\left<2\right>}],S[e_{i+5}^{\left<2\right>}]
\right)
\\
&=g^{(3,i)}\left(K_i,K_{i+1},K_{i+2},K_{i+4},K_{i+5},K_{i+8}\right).
\label{eq:main_bjf_3}
\end{align}
Similarly, we have  bijective functions $g^{(\mu,i)}$ for $\mu\in\{4,5,6\}$ and $i\in\mathbb Z/9\mathbb Z$
such that 
\begin{align}
&\nonumber\quad
\left(
S[e_i^{\left<1\right>}],S[e_{i+3}^{\left<1\right>}], 
S[e_{i+4}^{\left<3\right>}], S[e_{i+1}^{\left<3\right>}],
S[e_{i+2}^{\left<2\right>}],S[e_{i+5}^{\left<2\right>}]
\right)
\\
&=g^{(4,i)}\left(K_i,K_{i+1},K_{i+2},K_{i+3},K_{i+4},K_{i+5}\right)
\label{eq:main_bjf_4}
\\\nonumber\\
&\nonumber\quad
\left(
S[e_i^{\left<1\right>}],S[e_{i+3}^{\left<1\right>}], 
S[e_{i+4}^{\left<3\right>}], S[e_{i+1}^{\left<3\right>}],
\right.\\\nonumber &\quad\quad\quad\quad\left.
S[e_{i+6}^{\left<3\right>}], S[e_{i+2}^{\left<3\right>}],
S[e_{i+5}^{\left<2\right>}]
\right)
\\
&=g^{(5,i)}\left(K_i,K_{i+1},K_{i+2},K_{i+3},K_{i+4},K_{i+5},K_{i+6}\right)
\label{eq:main_bjf_5}
\\\nonumber\\
&\nonumber\quad
\left(
S[e_i^{\left<1\right>}],S[e_{i+2}^{\left<1\right>}], 
S[e_{i+5}^{\left<1\right>}], S[e_{i+6}^{\left<1\right>}],
S[e_{i+4}^{\left<3\right>}], S[e_{i+1}^{\left<3\right>}]
\right)
\\
&=g^{(6,i)}\left(K_{i},K_{i+1},K_{i+2},
K_{i+4},K_{i+5},K_{i+6}\right).
\label{eq:main_bjf_6}
\end{align}

From these relations, we see that, 
for any $e\in E_0^{\rm st}$, random variable $S[e]$ is generated by at least one of bijective functions on a subset of $\{K_i\}_{i\in\mathbb Z/9\mathbb Z}$, i.e. the eqs. 
(\ref{eq:main_bjf_1}),
(\ref{eq:main_bjf_2}),
(\ref{eq:main_bjf_3}),
(\ref{eq:main_bjf_4}), 
(\ref{eq:main_bjf_5}) and
(\ref{eq:main_bjf_6}).
This fact and the complete randomness of $\{K_i\}_i$ guarantee
that the $S[e]$ must have the maximum entropy,  i.e. 
\begin{align}
 H(S[e])=1.
\label{eq:main_en_z}
\end{align}
for  $e\in E_0^{\rm st}$.

The equations (\ref{eq:main_bjf_3})$\sim$(\ref{eq:main_bjf_6})
allow to write the random variable $S[e_i^{\left<1\right>}]$ in multiple expressions as follows:
\begin{align}
S[e_i^{\left<1\right>}]&=
g^{(4,i)}_1(K_i,K_{i+1},K_{i+2},K_{i+3},K_{i+4},K_{i+5})
\nonumber\\&=
g^{(4,i+6)}_2
(K_{i+6},K_{i+7},K_{i+8},K_i,K_{i+1},K_{i+2})
\nonumber\\&=
g^{(6,i+3)}_4
(K_{i+3},K_{i+4},K_{i+5},         K_{i+7},K_{i+8},K_i                  ).
\end{align}
Here, $g^{(\mu,i)}_k(\cdot)$ denotes the $k$-th element of the list of variables defined by the function  $g^{\mu,i}(\cdot)$.
From the complete randomness of $\{K_i\}_i$, $S[e_i^{\left<1\right>}]$ must depend only on the inter section of the sets as arguments of these functions:
\begin{align}
\{K_i,K_{i+1},K_{i+2},K_{i+3},K_{i+4},K_{i+5},\} 
&\nonumber\\
\cap
\{K_i,K_{i+1},K_{i+2},K_{i+6},K_{i+7},K_{i+8}\} 
&\nonumber\\
\cap
\{K_i,K_{i+3},K_{i+4},K_{i+5},K_{i+7},K_{i+8}\}
&=\{K_i\}.
\end{align}
This fact and the relation (\ref{eq:main_en_z}) imply that
\begin{align}
S[e_i^{\left<1\right>}]=K_i\oplus d[e_i^{\left<1\right>}]
\label{eq:main_bjf_7}
\end{align}
for a certain constant $d[e_i^{\left<1\right>}]$.

In the same way, from the multiple representations of $S[e_i^{\left<2\right>}]$ and $S[e_i^{\left<3\right>}]$:
\begin{align}
S[e_i^{\left<2\right>}]
&=
g^{(3,i+5)}_5(K_{i+5},K_{i+6},K_{i+7},K_{i},K_{i+1},K_{i+4})
\nonumber\\&=
g^{(3,i+4)}_6(K_{i+4},K_{i+5},K_{i+6},K_{i+8},K_{i},K_{i+3})
\nonumber\\&=
g^{(4,i+7)}_5(K_{i+7},K_{i+8},K_{i},K_{i+1},K_{i+2},K_{i+3}),
\\
\nonumber\\
S[e_i^{\left<3\right>}]
&=
g^{(4,i+8)}_4(K_{i+8},K_{i},K_{i+1},K_{i+2},K_{i+3},K_{i+4})
\nonumber\\&=
g^{(5,i+3)}_5
(K_{i+3},K_{i+4},K_{i+5},K_{i+6},K_{i+7},K_{i+8},K_{i})
\nonumber\\&=
g^{(6,i+5)}_5
(K_{i+5},K_{i+6},K_{i+7},K_{i},K_{i+1},K_{i+2}),
\end{align}
we can find that there are  certain constants $d[e_i^{\left<\gamma\right>}]$
such that the relation
\begin{align}
S[e_i^{\left<\gamma\right>}]&=K_{i}\oplus d[e_i^{\left<\gamma\right>}]
\label{eq:main_bjf_8}
\end{align}
hold for $\gamma\in\{2,3\}$.

Eqs.
 (\ref{eq:main_bjf_1}),
 (\ref{eq:main_bjf_2}),
 (\ref{eq:main_bjf_7}) and
 (\ref{eq:main_bjf_8})
prove the lemma.

\subsection{Proof of Lemma \ref{lm:cons_ord}}
\label{sec:order_on_standard_path}
Protocol $L_0$ can be viewed as a communication protocol performed by the subgraphs $G_{s,i}^{\rm bn}$ using the standard paths.


In this picture, the communication between $G_{s,i}^{\rm bn}$ satisfy the following properties.
\begin{itemize}
\item[1)] Subgraphs $G_{s,i}^{\rm bn}$ are connected solely by the standard paths.
\item[2)] Each standard path $i'$ is a straight line.
\item[3)] All edges in standard path $i'$ convey the same random bit $K_{i'}$, up to a constant (Lemma \ref{lm:cons_chan}).
\item[4)] Random bits $K_{i'}$ are independent of each other (due to the definition of the KRP).
\end{itemize}
For these properties to hold, it is necessary that
\begin{itemize}
\item[a)] Random bit $K_{i'}$ is generated inside one of subgraphs $G_{s,i}^{\rm bn}$ on the standard path $i'$.
\item[b)] Value of $K_{i'}$ thus generated is conveyed repeatedly to adjacent subgraphs $G_{s,i}^{\rm bn}$ on the same standard path $i'$.
\end{itemize}
and thus the lemma holds.

Indeed, if a) is not true, i.e., if $K_{i'}$ is generated independently by two or more of $G_{s,i}^{\rm bn}$, there is a nonzero probability that their values differ. For other subgraphs $G_{s,i}^{\rm bn}$ to be able to send out $K_{i'}$ thus generated, they must learn it from an adjacent subgraph which already knows $K_{i'}$.

\subsection{Proof of Lemma \ref{lm:cons_cont}}
\label{sec:proof_of_lemma_at_least_two_condition}
We will take two steps.

First, we will show that, for any order $\prec$, there is a subset $E'$ of $E_0^{\rm st}$
which satisfies the following four items:
\begin{itemize}
\item[1)] 
 the subset $E'$ contains the smallest edge in any standard path $i$.
%
\item[2)] $E_0^{\rm st}\neq E'$,
\item[3)] 
 When 
$i\in \mathbb Z/9\mathbb Z$ is a value satisfying  $e^{\left<\gamma\right>}_{i}, e^{\left<\gamma'\right>}_{i+1}\in E'$ for $\gamma\in \{1,2\}$ and $\gamma'\in \{0,1\}$,
the relation $e^{\left<3-\gamma\right>}_{i}, e^{\left<1-\gamma'\right>}_{i+1}\in E'$ holds.
\item[4)] 
When 
$i\in \mathbb Z/9\mathbb Z$ is a value satisfying  $e^{\left<\gamma\right>}_{i},e^{\left<\gamma'\right>}_{i+5}\in E'$ for 
$\gamma\in \{3,4\}$ and 
 $\gamma'\in \{2,3\}$, the relation 
$e^{\left<7-\gamma\right>}_{i}, e^{\left<5-\gamma'\right>}_{i+5}\in E'$ holds.
\end{itemize}

Second, 
we will show  that the existence of $E'$ defined above  and conditions $C_1$ and $C_2$ are incompatible for the order $\prec$ defined from a secure protocol $L_0$.

Thus, Lemma  Lemma \ref{lm:cons_cont} holds.

The first step is shown as follows:
 We constructively show that for any sequence $\{\gamma_i\in\{0,1,2,3,4\}\}_{i\in \mathbb Z/9\mathbb Z}$,
there is a subset which satisfies 
\begin{itemize}
\item[1')] $\forall i,\quad e^{\left<\gamma_i\right>}_{i}\in E'$,
\end{itemize}
2), 3), and 4).
Such a subset is one of 
 $E^{(1)}$,
 $E^{(2)}$,
 $E^{(3)}_{i,i'}$ and 
 $E^{(4)}_{q}$ for $i\neq i'\in\mathbb Z/9\mathbb Z$ and $q\in\{0,1,2\}$ defined below:
\begin{align}
E^{(1)}&:=\{e^{\left<\gamma\right>}_i|
i\in \mathbb Z/9\mathbb Z,\quad
\gamma\in\{2,3,4\}
\}\\
E^{(2)}&:=\{e^{\left<\gamma\right>}_i|
i\in \mathbb Z/9\mathbb Z,\quad
\gamma\in\{0,1,2\}
\}
\\
E^{(3)}_{i',i''}&:=
\left\{e^{\left<\gamma\right>}_i\left|
\begin{array}{l}
\;
(i=i',  \gamma\in\{0,1\})
\\\lor
(i\in \{i'+1,i'+2,\cdots, i''-1\},
\\\quad\quad
 \gamma=2)
\\\lor
(i=i'', \gamma\in\{3,4\})
\\\lor
(i\in \{i''+1,i''+2,\cdots, i'+4\},
\\\quad\quad
 \gamma\in\{0,1,2,3,4\})
\\\lor
(i\in \{i'+5,i'+6,\cdots, i''+4\},
\\\quad\quad
 \gamma\in\{0,1,2\})
\\\lor(i\in \{i''+5,i''+6,\cdots, i'+8\},
\\\quad\quad
 \gamma\in\{0,1,2,3,4\})
\end{array}
\right.\right\}
\\
E^{(3)}_{i'',i'}&:=
\left\{e^{\left<\gamma\right>}_i\left|
\begin{array}{l}
\;
(i=i',  \gamma\in\{3,4\})
\\\lor
(i\in \{i'+1,i'+2,\cdots, i''-1\},
\\\quad\quad
 \gamma\in\{0,1,2,3,4\})
\\\lor
(i=i'',  \gamma\in\{0,1\})
\\\lor
(i\in \{i''+1,i''+2,\cdots, j'+4\},
\\\quad\quad
 \gamma=2)
\\\lor
(i\in \{i'+5,i'+6,\cdots, i''+4\},
\\\quad\quad
 \gamma\in\{2,3,4\})
\\
\lor(i\in \{i''+5,i''+6,\cdots,i'+8\},
\\\quad\quad
\gamma=2)
\end{array}
\right.\right\}
\\
E^{(4)}_{q}&:=
\left\{e^{\left<\gamma\right>}_i\left|
\begin{array}{cll}
&
(\makebox[1.6cm]{$i\equiv q$}\;{\rm mod}\: 3,\; \gamma\in\{0,1\})
\\
\land
&
(\makebox[1.6cm]{$i\equiv q+1$}\;{\rm mod}\: 3,\; \gamma\in\{2,3,4\})
\\
\land
&
(\makebox[1.6cm]{$i\equiv q+2$}\;{\rm mod}\: 3,\; \gamma\in\{3,4\})
\end{array}
\right.\right\}
\end{align}
where $i''\in\{i'+1,i'+2,i'+3,i'+4\} \subset \mathbb Z/9\mathbb Z$.
We can straightforwardly check that all the subsets satisfy the last three items 2), 3), and 4) directly.
For any sequence  $\{\gamma_j\in\{0,1,2,3,4\}\}_{j\in \mathbb Z/9\mathbb Z}$,
we can find that one of the above subsets satisfies the item  1') as well.
We confirmed this fact by a brute force search using a computer.
This completes the first step of the proof.

For the second step, for any given order $\prec$ defined from a secure protocol $L_0$, we will derive a contradiction from the assumptions that 
conditions C1 and C2 hold and that there exists $E'$ which satisfies the four items 1), 2), 3), and 4).
 We pick up the smallest edge $e^{\left<\gamma_m\right>}_{i_m}$  in $E_0^{\rm st}\backslash E' $.
Existence of it guaranteed from the item 2).
From the first item 1) and Lemma \ref{lm:cons_ord}, there is a node $e^{\left<\gamma'\right>}_{i_m}\in E'$ such that $|\gamma'-\gamma_m|=1$.

When $\gamma', \gamma_m\in\{0,1\}$, we will give a contradiction, as an example.
 The third item 3) enforces 
\begin{align}
e^{\left<1\right>}_{i_m+8}, e^{\left<2\right>}_{i_m+8} \in E_0^{\rm st}\backslash E' .
\label{eq:tmp_10}
\end{align} 
From this relation and the minimality of $e^{\left<\gamma_m\right>}_{i_m}$  in $E_0^{\rm st}\backslash E' $,
 the order 
$e^{\left<\gamma'\right>}_{i_m}\prec e^{\left<\gamma_m\right>}_{i_m}\prec e^{\left<1\right>}_{i_m+8}, e^{\left<2\right>}_{i_m+8}$ must hold. 
From this order, C1 enforces us that $e^{\left<1\right>}_{i_m+8}$ or $e^{\left<2\right>}_{i_m+8}$  must be a smallest edges in the standard path $i_m+8$. Here, we have use the facts that 
$e^{\left<\gamma'\right>}_{i_m}=e^{\left(\gamma'+2\right)}_{1,i_m+8}$, $e^{\left<\gamma_m\right>}_{i_m}=e^{\left(\gamma_m+2\right)}_{1,i_m+8}$, $e^{\left<1\right>}_{i_m+8}=e^{\left(0\right)}_{1,i_m+8}$, and $e^{\left<2\right>}_{i_m+8}=e^{\left(1\right)}_{1,i_m+8}$. As a result, the item 1) enforces us that 
 $e^{\left<1\right>}_{i_m+8}$ or  $e^{\left<2\right>}_{i_m+8}$ must be in $E'$.
 However, this relation contradicts the relation (\ref{eq:tmp_10}).
 
In the same way, we can derive contradictions in the other cases, i.e. 
$\gamma',\gamma_m\in\{1,2\}$ or $\gamma', \gamma_m\in\{2,3\}$ or $\gamma', \gamma_m\in\{3,4\}$.

\section{Summary and outlook}
We investigated relations between the key relay protocol (KRP) and secure network coding (SNC) under the one-shot scenario, and also under the scenario where wiretap sets are restricted.
We found that there is a definite gap in security between these two types of protocols; namely, certain KRPs achieve better security than any SNC schemes on the same graph.
We also found that this gap can be closed by generalizing the notion of SNC by adding free public channels;
that is the KRP is equivalent to SNC augmented with free public channels.

There are still many open problems.
For example, does the gap we found here persist even under the asymptotic case?

It is also interesting to figure out on what types of graphs the gap occurs.
Our conjecture is that there is no gap on plane graphs, and also for the case where there is only one sender-receiver pair, though the rigorous proofs remain as future works.

\appendix
\section*{Formal proof of Lemma \ref{lmm:security_L}}
\begin{proof}
We prove the lemma by using a little bit modified protocol $L_{\rm KRP}^{\rm m}$, such that 
 $L_{\rm KRP}$ being secure is equivalent to  $L_{\rm KRP}^{\rm m}$ being secure.

$L_{\rm KRP}^{\rm m}$
is obtained    by the following three process.
First,
when a public message made at any node in $L_{\rm KRP}$ can be expressed as 
 a linear combination of the other  public messages $p_x$ and the local keys $r_y$ held by the node, i.e. $\bigoplus_x p_x\oplus \bigoplus_y r_y$, the corresponding message  made at the node in $L_{\rm KRP}^{\rm m}$ is  a linear combination of the local keys $r_y$ only, i.e.  the parity of them $\bigoplus_y r_y$.
Second, all the public communications in $L_{\rm KRP}^{\rm m}$ are used only for sending the parities to all users.
Finally, as relayed keys, the users evaluate the same values as for $L_{\rm KRP}$. 

From this relation, we know that 
any bit obtained by the adversary in the case of $L_{\rm KRP}^{\rm m}$ can be evaluated
  by the adversary in the case of $L_{\rm KRP}$, and vice versa.
This is why $L_{\rm KRP}$ being secure is equivalent to  $L_{\rm KRP}^{\rm m}$ being secure.

Self-contained definition of $L_{\rm KRP}^{\rm m}$ is as follows: 
$L_{\rm KRP}^{\rm m}$  is the protocol in which the following four phases are implemented in sequence.
\begin{itemize}
\item 
{\it Local key generation phase:}
 All channels $LKS_e$ are used to generate local keys.
\item 
 {\it Parity evaluation phase:}
 Each node evaluates parities of (part of) local keys held by the node.
\item 
 {\it Public communication  phase:}
 Evaluated parities are transferred  from each node to users via public communications.
\item 
 {\it Relayed key generation  phase:}
 Each user generates the relayed key from received local keys and the parities.
\end{itemize}
{\it Parity evaluation phase} and 
 {\it Relayed key generation  phase} is explicitly identified by the following definitions:
\begin{itemize}
\item {\it All the parities each nodes evaluate:}
In each sub-graph $G_{s,i}^{\rm bn}$,
all the parities each nodes  $v^{(\alpha)}$ evaluate are
\begin{align}
v^{(0)}:&&
p[v^{(0)},1]&:=r[e^{(0)}]\oplus r[e^{(4)}]
\label{def:v_0_1}
\\
&&
p[v^{(0)},2]&:=r[e^{(4)}]\oplus r[e^{(6)}]
\\
v^{(1)}:&&
p[v^{(1)}]&:= r[e^{(8)}]\oplus r[e^{(10)}]\oplus r[e^{(1)}]
\\
v^{(2)}:&&
p[v^{(2)},1]&:=r[e^{(5)}]\oplus r[e^{(8)}]
\\
&&
p[v^{(2)},2]&:=r[e^{(2)}]\oplus r[e^{(5)}]
\\
v^{(3)}:&&
p[v^{(3)}]&:= r[e^{(6)}]\oplus r[e^{(9)}]\oplus r[e^{(3)}]
\\
v^{(4)}:&&
p[v^{(4)}]&:=r[e^{(4)}]\oplus r[e^{(5)}]
\oplus r[e^{(7)}]
\\
v^{(5)}:&&
p[v^{(5)},1]&:=r[e^{(7)}]\oplus r[e^{(10)}]
\\
&&
p[v^{(5)},2]&:=r[e^{(7)}]\oplus r[e^{(9)}].
\label{def:v_5_2}
\end{align}
In order to simplify the next expression, we introduce notations:
\begin{align}
p_{s,i}^{(0)}&:=
 p[v^{(0)},1]
\oplus p[v^{(4)}]
\oplus p[v^{(2)},1]
\label{def:p_1}
\nonumber\\
&\quad\quad
\oplus p[v^{(5)},1],
\oplus p[v^{(1)}]
\\
\label{def:p_2}
p_{s,i}^{(1)}&:=
          p[v^{(2)},2]
\oplus p[v^{(4)}]
\oplus p[v^{(0)},2]
\nonumber\\
&\quad\quad
\oplus p[v^{(5)},2]
\oplus p[v^{(3)}]
\end{align}

\item {\it The function which gives the relayed keys:}
The relayed keys generated by  $u_i^1$ and $u_i^2$ are
\begin{align}
k_i^1&:=r[e_{1,i+8}^{(2)}]
\label{def:k_1}
\\
k_i^2&:=
 p_{2,2i}^{(1)}
\oplus p_{1,i+8}^{(1)}
\oplus p_{1,i}^{(0)}
\oplus p_{2,2i+1}^{(0)}
\oplus r[e_{2,2i+1}^{(1)}]
\label{def:k_2}
\end{align}
for $i\in\mathbb Z/9\mathbb Z$.

\end{itemize}

From the definitions (\ref{def:v_0_1}),$\dots$,(\ref{def:p_2}),
\begin{align}
p_{s,i}^{(0)}&=r[e^{(0)}_{s,i}]\oplus r[e_{s,i}^{(1)}]
\label{eq:p_1}
\\ 
p_{s,i}^{(1)}&=r[e^{(2)}_{s,i}]\oplus r[e^{(3)}_{s,i}]
\label{eq:p_2}
\end{align}
is obtained 
for
each sub-graph $G_{s,i}^{\rm bn}$.
Here, we have used the fact that $r\oplus r=0$ for $r\in\{0,1\}$.
By using the relations (\ref{eq:p_1}) and (\ref{eq:p_2}), 
the relayed keys (\ref{def:k_1}) and (\ref{def:k_2}) are evaluated as
\begin{align}
k_i^1=k_i^2=r[e_{1,i+8}^{(2)}].
\end{align}
 This relation implies the soundness of  $L_{\rm KRP}$.
Since the generated relayed keys are part of the local keys, and the public information is linear combinations of the local keys,
the secrecy can be checked from the fact that the relayed keys are linearly independent of all the published  information.
\end{proof}


\section*{Proof of Lemma\ref{th:main_2}}

From the assumption (\ref{ass:th1_1})
\begin{align}
I(B,C_0;B,C_1)\geq I(A;B,C_1)\geq I(A;A)=H(A).
\end{align}
This relatoin and the assumption (\ref{ass:th1_2}),
guarantee that 
\begin{align}
I(B,C_0;B,C_1)&= I(A;B,C_1),
\\
H(A)&=H(B)
\label{eq:th1_0}
\end{align}
The first relation and the assumption (\ref{ass:th1_1}) imply 
\begin{align}
&P(A=a)P(A=a,B=b,C_0=c_0,C_1=c_1)
\nonumber\\
=& P(A=a,B=b,C_0=c_0)P(A=a,B=b,C_1=c_1)
\label{eq:th1_1}
\end{align}
for any $a,b,c_0,c_1$. Here we have used 
the fact that, 
if $I(Z;Y)= I(X;Y)$ for $Z:=f(X)$,  the relation $\forall x,y,z,\;
P(X=x,Y=y,Z=z)P(Z=z)=P(X=x,Z=z)P(Y=y,Z=z)$ holds.
By summing up with respect to $c_1$,
the relation becomes 
\begin{align}
&P(A=a,B=b,C_0=c_0)P(A=a)
\nonumber\\
=& P(A=a,B=b,C_0=c_0)P(A=a,B=b).
\label{eq:th1_2}
\end{align}
Since $A=f_0(B,C_0)$,
we can define functions $h_0$, and $h_1$ such that 
 $a=f_0(h_0(a),h_1(a))$ and $P(A=a,B=h_0(a),C_0=h_1(a))\neq 0$ hold, if $P(A=a)\neq 0$. Using these functions, by substitute $h_0(a)$ and $h_1(a)$ into $b$ and $c_0$ respectively in the above relation,
 \begin{align}
P(A=a)= P(A=a,B=h_0(a)).
\label{eq:th1_3}
\end{align}
is obtained.
This relation guarantees the following relation
\begin{align}
&\sum_{a,b} P(A=a,B=b)\delta(b,h_0(a))
\nonumber\\
&=
\sum_{a} P(A=a,B=h_0(a))
\nonumber\\
&=
\sum_{a} P(A=a)=1.
\end{align}
i.e. $h_0(A)=B$.
Therefore,
\begin{align}
I(A;B)=H(B)=H(A)
\end{align}
holds where the eq.(\ref{eq:th1_0}) is used in the second equality.
As a result, there is a function $h_2$ such that $h_2(B)=A$.
The last assumption, i.e. the number of candidates for $A$ is equal to that for $B$, and
the existence of the functions $h_0$ and $h_2$ guarantee the existence of the bijective function $g$ such that $g(A)=B$.

\section*{Acknowledgment}
G.K. was supported in part by JSPS Kakenhi (C) No. 20K03779 and 21K03388.
M.F. and T.T. were supported in part by ``ICT Priority Technology Research and Development Project'' (JPMI00316) of the Ministry of Internal Affairs and Communications, Japan.

\bibliographystyle{IEEEtran}
\bibliography{KR_SNC}

\begin{IEEEbiographynophoto}{Go Kato} was born in Japan, in 1976. He received the M.S. and Ph.D.  
degrees in science from The University of Tokyo in 2001 and 2004, respectively.
In 2004, he joined the NTT Communication Science Laboratories and has been engaged in the theoretical investigation of quantum information. He is especially interested in mathematical structures emerging in the field of quantum information. He is a member of the Physical Society of Japan.
\end{IEEEbiographynophoto}

\begin{IEEEbiographynophoto}{Mikio Fujiwara} received the B.S. and M.S. degrees in electrical engineering and the Ph.D. degree in physics from Nagoya University, Nagoya, Japan, in 1990, 1992, and 2002, respectively. He has been involved R\&D activities at NICT (previous name CRL,  Ministry of Posts and Telecommunications of Japan) since 1992.
\end{IEEEbiographynophoto}

\begin{IEEEbiographynophoto}{Toyohiro Tsurumaru} was born in Japan in 1973.
He received the B.S. degree from the Faculty of Science, University of Tokyo, Japan in 1996,
and M.S. and Ph.D. degrees in physics from the Graduate School of Science, University of Tokyo, Japan in 1998 and 2001, respectively.
Then he joined Mitsubishi Electric Corporation in 2001.
His research interests include theoretical aspects of quantum cryptography and of modern cryptography.
\end{IEEEbiographynophoto}

\end{document}